\documentclass[journal]{IEEEtran}

\usepackage[pdftex]{graphicx}
\usepackage{multirow}

\usepackage{amsthm}

\usepackage{array}
\newtheorem{theorem}{Theorem}

\newtheorem{proposition}{Proposition}

\newtheorem{definition}{Definition}
\usepackage{epstopdf}
\usepackage{ amssymb }
\usepackage{amsmath}

\ifCLASSINFOpdf
\else
\fi
\begin{document}
%
\title{Joint Caching and Pricing Strategies for Popular Content in Information Centric Networks}
%
%
%

\author{Mohammad~Hajimirsadeghi,~\IEEEmembership{Student Member,~IEEE,}
        Narayan B.~Mandayam,~\IEEEmembership{Fellow,~IEEE,}
        and~Alex~Reznik,~\IEEEmembership{Senior Member,~IEEE}
\thanks{Manuscript received April 30, 2016; revised September 29, 2016 and December 14, 2016; accepted
January 23, 2017.}
\thanks{M. Hajimirsadeghi and N. B. Mandayam are with WINLAB, Rutgers University, Piscataway, NJ 08854-8060 USA (e-mail: narayan@winlab.rutgers.edu; mohammad@winlab.rutgers.edu).}
\thanks{A. Reznik is with Hewlett Packard Enterprise (HPE), New York 10011, USA (e-mail: alex.reznik@hpe.com ).}}

%
%

\pagestyle{plain}


%



\maketitle

\begin{abstract}
We develop an analytical framework for distribution of popular content in an Information Centric Network (ICN) that comprises of Access ICNs, a Transit ICN and a Content Provider. Using a generalized Zipf distribution to model content popularity, we devise a game theoretic approach to jointly determine caching and pricing strategies in such an ICN. Under the assumption that the caching cost of the access and transit ICNs is inversely proportional to popularity, we show that the Nash caching strategies in the ICN are 0-1 (all or nothing) strategies. Further, for the case of symmetric Access ICNs, we show that the Nash equilibrium is unique and the caching policy (0 or 1) is determined by a threshold on the popularity of the content (reflected by the Zipf probability metric), i.e., all content more popular than the threshold value is cached. We also show that the resulting threshold of the Access and Transit ICNs, as well as all prices can be obtained by a decomposition of the joint caching and pricing problem into two independent caching only and pricing only problems.   
\end{abstract}

\thispagestyle{empty}

%
\IEEEpeerreviewmaketitle

\section{Introduction}
%
%
%
%


\IEEEPARstart{T}{he} vast majority of Internet traffic relates to content access from the sources such as YouTube, Netflix, Bit Torrent, Hulu, etc. This rapid increase of content delivery in the Internet has revealed the need for a different networking paradigm. Further, as Fig. \ref{fig:fig1} describes, the emerging trend is that the users are just interested in the information (content), and not where it is located or perhaps, even how it is delivered. This high increase in demand for video content in the Internet and the need for new approaches to control this large volume of information have motivated the development of future Internet architectures based on named data objects (NDOs) instead of named hosts \cite{Jacobson}. Such architectural proposals are generally referred to as Information Centric Networking (ICN) which is a new communication paradigm to increase the efficiency of content delivery and also content availability \cite{Ahlgren}-\cite{Trossen} of future fifth generation (5G) networks. In this new concept, the network infrastructure actively contributes to content caching and distribution and every ICN node can cache and serve the requested content. To fulfill that purpose, several architectures have been proposed for ICN to reflect current and future needs better than the existing Internet architecture \cite{Dannewitz}-\cite{Trossen2}. To provide preferable services to the users in ICN, Internet service providers (ISPs) or access ICNs should be able to maintain quality of service (QoS) by improving the response time for file request. They need to cache the frequently requested or popular content locally and store them near the users in the network. To provide QoS, in-network caching is introduced to provide the network components with caching ability. Therefore every node actively contributes in content caching and operates as a potential source of content. This leads to reduction in network congestion and user access latency and increase in the throughput of the network by locally caching the more popular content \cite{Zhang}-\cite{Suksomboon}. Several works have claimed that web (file) requests in Internet are distributed according to Zipf’s law \cite{Adamic}-\cite{Shi}. Zipf’s law states that the relative probability of a request for the $i$th most popular content is proportional to $\frac{1}{i}$. However, several other studies have found out that the request distribution generally follows generalized Zipf distribution where the request rate for the $i$th most popular content is proportional to $\frac{1}{{{i^\gamma }}}$ and $\gamma$ is a positive value less than unity \cite{Breslau}-\cite{Fricker}.


\begin{figure}
\hspace*{-1.7cm} 
  \includegraphics[width=11.5cm,height=9cm]{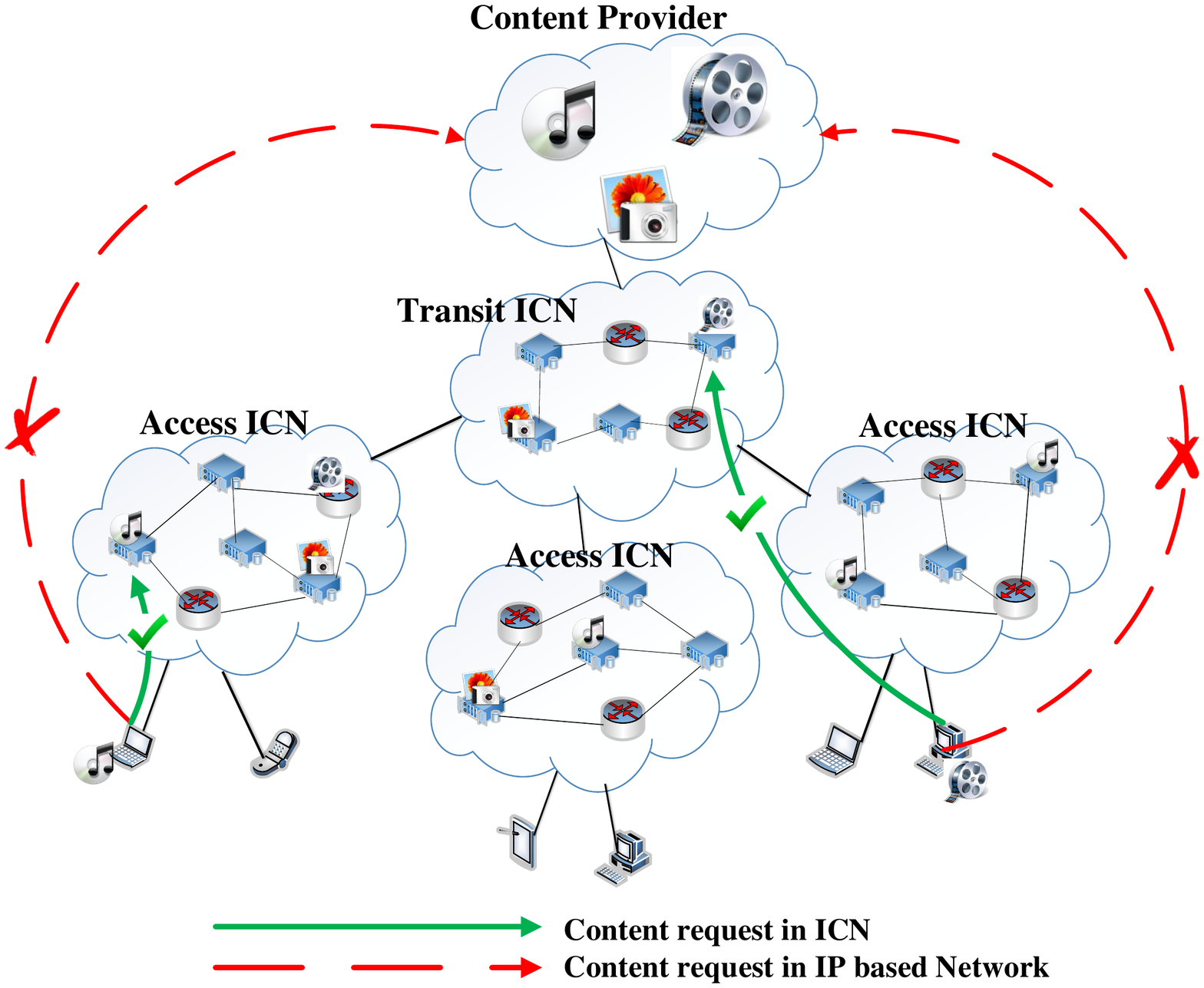}
  \caption{ICN communication model: unlike the current Internet which the users are interested in the location of the files, the ICN users just look for the content regardless of file location \cite{Hajimirsadeghi}.}
  \label{fig:fig1}
\end{figure}

Since each ICN requires cooperation in caching from other ICNs to provide a global high performance network, it is necessary to have pricing policies to incentivize all the ICNs to contribute to the caching process  \cite{Clark}. Several works have been done to address the problem of the economics of service pricing in current Internet and interconnection networks \cite{Shakkottai}-\cite{Wu}. Using contemporary pricing policies cannot incentivize the lower tier ISPs to cooperate in the future Internet architecture \cite{Rajahalme}; hence, we need to have new models to provide them with monetary incentives to collaborate in caching and distributing content when content with different popularities are available in the network.

In this paper, we investigate joint caching and pricing strategies of the access ICN, the transit ICN and the content provider based on content popularity. We study Nash strategies for a non-cooperative game among the above entities using a probabilistic model by assuming that access requests generally follow the generalized Zipf distribution. We then use the insights gained to simplify the problem by replacing two caching threshold indices instead of caching parameters for the symmetric case; where all access ICNs have the same parameters. In our model, the ICNs’ caching costs vary in respect to content popularity while the content provider cost per unit data is fixed for all content types.

The rest of the paper is organized as follows. Section II reviews some related works. In Section III, we describe the system model. The problem formulation for joint caching and pricing strategies is described in Section IV. We examine the generalized scenario for arbitrary number of access ICNs in 
Section V. The symmetric access ICNs scenario is analyzed in Section VI. Section VII presents numerical results on the resulting caching and pricing strategies, and we conclude in Section VIII.


\section{Related Work }

The benefits of in-network caching have been investigated before in the setting of distributed file systems in several recent works. In \cite{Maddah}, the problem of caching is studied from an information-theoretic viewpoint. They propose a coded caching approach that in addition to the local caching gain is able to achieve a global caching gain. A novel Cooperative Hierarchical Caching (CHC) framework is proposed in \cite{Tran}, \cite{Tran2} in the context of Cloud Radio Access Networks (C-RAN). In \cite{Tran3}, a collaborative joint caching and processing strategy for on-demand video streaming in mobile-edge computing network is envisioned.Content caching and delivery in device-to-device (D2D) networks have been studied in \cite{Ji}. The aim of this work is to improve the performance of content distribution by the use of caching and content reuse. Several approaches such as base station assisted D2D network and other schemes based on caching at the user device are compared to show the improvement of the network throughput in the presence of in-network caching. Another recent article \cite{Bastug} studies the limitation of current reactive networks and proposes a novel proactive networking paradigm where caching plays a crucial role. It shows that peak data traffic demands can be substantially reduced by proactively serving users' demands via caching. \cite{Hamidouche} has used a mean field game model to study distributed caching in ultra dense small cell networks. Zhang et al have proposed an optimal cache placement strategy based on content popularity in content centric networks (CCN) in \cite{Zhang2}. The authors in \cite{Guo} have proposed a collaborative caching and forwarding design for CCN. The problem of joint caching and pricing for data service for a single Internet service provider (ISP) is studied in \cite{Tadrous}. Similar problem but for multiple ISPs in the setting of small cell networks is investigated in \cite{Li2} using a Stackelberg game. \cite{Shen} proposes an incentive proactive cache mechanism in cache-enabled small cell networks (SCNs) in order to motivate the content providers to participate in the caching procedure.

One of the earliest studies of economic incentives in ICNs has been conducted by Rajahalme et al. \cite{Rajahalme} and has demonstrated that top level providers are not willing to cooperate in the caching process since they cannot get enough revenue. Another recent study by Agyapong et al. \cite{Agyapong} has addressed the economic incentive problem in ICN by building a simple economic model to evaluate the incentive of different types of network players in a hierarchical network infrastructure. They qualitatively showed that without explicit monetary compensation from publishers, the network will fail to deploy the socially optimal number of caches. Few prior works have used game theoretic approaches to solve the problem of caching and pricing in ICN. In \cite{Arifuzzaman} the authors have presented a game theoretic approach using matrix payoff to analyze the process of economic incentives sharing among the major network components. A pricing model was proposed in \cite{Pham} to study the economic incentives for caching and sharing content in ICNs which consists of access ICNs, a transit ICN and a content provider. This work has shown that if each player's caching (pricing) strategy remains fixed, the utility of each player becomes a concave function of its own pricing (caching) strategy. Therefore a unique Nash Equilibrium exists in a non-cooperative pricing (caching) game among different players. In our earlier work \cite{Hajimirsadeghi}, a similar model was adopted to address the problem of joint caching and pricing strategies in a network including two access ICNs, one transit ICN and one content provider. However, content popularity was not taken into account and the ICNs were agnostic to content type. The interaction of the above entities is modeled using game theoretic approach. It was shown that each player can optimize its caching and pricing strategies in a non-cooperative game. At the Nash Equilibrium, the caching strategies turn out to be 0-1 (all or nothing) strategies, where each access ICN caches all the requested demand if its caching cost is less than the transit ICN’s caching cost. When the caching cost of the access ICN is higher than the transit ICN's, all the requested content will be served by either the transit ICN or the content provider, whichever has the lower caching cost. It means that the content would be cached in the ICN with smallest caching cost. To the best of our knowledge, none of the earlier works consider the problem of joint caching and pricing in an ICN in the presence of content with different popularity. In contrast, in this paper, we study joint caching  and pricing strategies with the notion of content popularity for an ICN network with different elements considering asymmetric utility functions.


\section{System Model}

\subsection{Content Popularity}
There are $M$ different types of content in the network that each user is trying to access. Each type of content has a different measure of popularity reflected by the probability of requests for it. We consider a model where the popularity of content is uniformly similar in all parts of the network, i.e., all users in the network have the same file popularity distribution. Analyzing the impact of different per-user file popularities is an open problem. As in previous works (e.g. \cite{Li}-\cite{Adamic}, \cite{Shi}-\cite{Fricker}, \cite{Ji}-\cite{Guo}), in this paper the distribution of user requests for content is described by a generalized Zipf distribution function as follows:

\begin{equation} \label{eq:1}
{q_M}\left( m \right) = \frac{\Omega }{{{m^\gamma }}},m = 1,...,M,
\end{equation}
where $\Omega = {\left( {\sum\limits_{i = 1}^M {\frac{1}{{{i^\gamma }}}} } \right)^{ - 1}}$ and  $0 \le \gamma  \le 1$ is the exponent characterizing the Zipf distribution in which $\gamma = 0$ makes the distribution uniform and all the content will be identical in popularity, whereas the case of $\gamma  = 1$ corresponds to one where the content popularity distribution is following the classic Zipf's law and more popular content is dominant in the network. The content is ranked in order of their popularity where content $m$ is the $m$th most popular content, i.e., $m=1$ is the most popular content and $m=M$ is the least popular content.

In most of the aforementioned works that have studied in-network caching, the popularity profile of content was assumed to be identical and perfectly known by all the network components. In reality, the demand and popularity are not predictable and certain \cite{Bastug2}, \cite{Tadrous1}. The problem of caching has been studied in \cite{Blasco} wherein the users have access to demand history but no knowledge about popularity. Several other papers have used learning-based approaches to estimate the popularity profile at the user side \cite{Sengupta}-\cite{Bastug4}. While content learning is more accurate for modeling content popularity, the reason we have used the Zipf model is due to (1) experimental results showing reasonable fit to the Zipf model and (2) analytical tractability provided by the Zipf model. Our framework can be extended by changing the demand model and considering a repeated game with a parametric Zipf distribution. In each time slot of the game, this parameter can be estimated in an optimal way using a learning process.


\subsection{Cost Model}
Although the prices are fixed for all types of content, since the ICNs want to earn more profit by caching the content, they are more willing to locally store the content which is more popular. Thus, the access ICNs' and the transit ICN's caches treat the content differently in regard to their popularity. As the content gets more popular, the ICNs incur less caching costs to locally store the content.

\begin{definition}
\label{t1}
For a finite cache, the caching costs of access and transit ICNs is defined to be inversely proportional to the content popularity as follows
\begin{equation} \label{eq:2}
{c_{{x_i}}} = \frac{{{c_{{x_0}}}}}{{{q_M}\left( i \right)}},
\end{equation}
where, without loss of generality, ${c_{{x_0}}}$ is a fixed initial caching cost at ICN $x$.
\end{definition}

Using equations (\ref{eq:1}) and (\ref{eq:2}), we see in Fig. \ref{fig:accessC2}, for fixed values of $i$, ${c_{{x_i}}}$ is a decreasing function of $\gamma$ when $i$ is small (i.e., more popular content). On the other hand, ${c_{{x_i}}}$ is an increasing function of $\gamma$ when $i$ is large (i.e., less popular content). Unlike the access and transit ICNs, the content provider has no priority for caching the content and caches all types of content. The content provider incurs the constant cost $c_O$ for every unit of data that it serves for the transit ICN.

%


\subsection{Access and Demand Model}
For simplicity in illustration, as shown in Fig. \ref{fig:fig2}, we begin with a hierarchical network model \cite{Pham}, \cite{Borst} with two access ICNs (A and B), one transit ICN (C), one content provider (O) and an arbitrary number of users who can switch from one access ICN to another. The access ICNs connect the end-users to the content network and the transit ICN provides wide-area transport for the access ICNs while the content provider provides the content for the users. Fig. \ref{fig:fig2} also shows the monetary and data flows among different entities with the various prices described in Table \ref{table1}. The network economy depends on two effective factors: caching and pricing. Under the assumption that each ICN can have access to all content, it can decide to either cache the entire or portion of the requested content, or get it from somewhere else. The caching strategy adopted by each entity is denoted by the parameter $\alpha$  that takes values in the interval $\left[ {0,1} \right]$. Every ICN decides to cache different types of content independently, therefore, we have a specific caching variable for each type of content. We have denoted ${\alpha _{I,{S_i}}}$, where $I$ can either be  access or transit ICN and $S$ is any cache in the network (possibly another ICN or content provider), as the fraction of ICN $I$'s demand for content type $i$ that comes from cache $S$. Each ICN also has different pricing strategies. These strategies are the prices that each player charges others for the provided service. The pricing is based on the usage, i.e., price per unit data. Each access ICN sets two different prices: (1) the network price per unit data for transporting the content; and (2) the storage price per unit data for providing content from its cache for other ICNs. For example, the network and storage prices for access ICN $A$ are denoted as $P_A^{\left( n \right)}$ and $P_A^{\left( s \right)}$, respectively. The total price per unit data is the sum of these two prices and is denoted by ${P_A} = P_A^{\left( n \right)} + P_A^{\left( s \right)}$. We will find it useful to utilize an alternative form of the above. Following the charging policies of several storage services for the Internet such as Amazon S3, it would be practical to assume that the storage price is typically less than the network price. The linear relationship between network price and storage price while empirical, has been used earlier in \cite{Pham} and we follow this assumption. It can be represented by $P_A^{\left( n \right)} = {\beta _A}P_A^{\left( s \right)}$ where ${\beta _A} > 1$. Thus, the relationship between ${P_A}$ and $P_A^{\left( s \right)}$ would be $P_A^{\left( s \right)} = \left( {\frac{1}{{1 + {\beta _A}}}} \right){P_A}$. As a result, each access ICN or transit ICN will have a set of strategies for pricing in the interval $\left[ {0,\infty } \right)$. The content provider pricing strategies set also consists of the content price $P_O^{\left( c \right)}$ that the users should pay for content and the storage price $P_O^{\left( s \right)}$ which is the price for providing the content from the content provider cache. The access ICN $A$ and $B$ charge prices ${P_A}$ and $P_B$ to their users and $P_A^{\left( s \right)}$ and $P_A^{\left( s \right)}$ to the transit network if they store or forward the content that the transit ICN had asked for. The transit ICN $C$ charges access ICNs $A$ and $B$ with price ${P_C}$, if it stores or forwards their requested content. 
The content provider charges users with content price $P_O^{\left( c \right)}$ and transit ICN $C$ with storage price $P_O^{\left( s \right)}$  if it stores transit ICN requested content.


\begin{figure}
\centering
\hspace*{-0.5cm}
  \includegraphics[width=11cm,height=9cm]{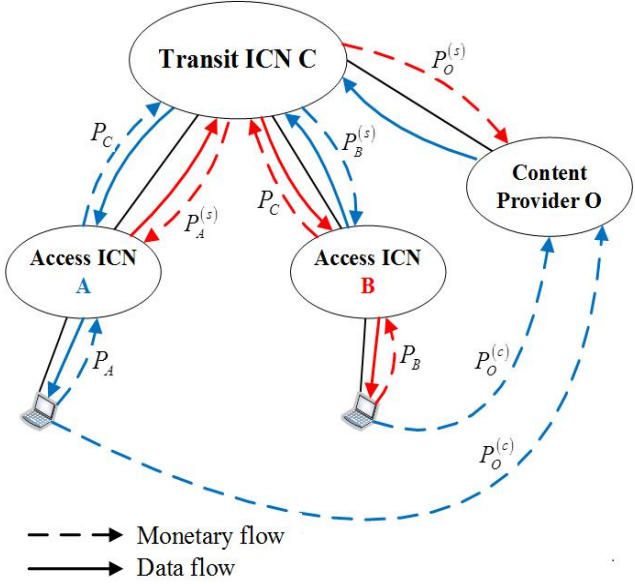}
  \caption{Interaction between different entities in a simplified model of an ICN \cite{Hajimirsadeghi}.}
\label{fig:fig2}
\end{figure}


To model the behavior of users, we have considered content demand at each access ICN to be a linear function of the prices.

\begin{definition}
\label{t2}
The users' demands are affected by both content price and access price and defined to be a linear function as follows:  

\begin{equation} \label{eq:3}
\begin{array}{l}
{\sigma _A} = 1 - {\rho _A}{P_A} + {\rho _B}{P_B} - {\rho _0}P_o^{(c)}\\
{\sigma _B} = 1 + {\rho _A}{P_A} - {\rho _B}{P_B} - {\rho _0}P_o^{(c)}

\end{array}
\end{equation}
\end{definition}
\noindent where ${\rho _A}$, ${\rho _B}$ and ${\rho _0}$ are the reflective coefficients of the prices' influence on users’ demands and are positive. The demands are normalized per unit data. We observe that the users’ demands directly depend on the access ICNs prices and content price. For example, as the content price $P_O^{\left( c \right)}$  increases, the users’ demands decrease. If one of the access ICNs increases its price, the users will switch to the other ICN. Table \ref{table1} summarizes our notation. 

%
%
%



\begin{table}
\centering
\caption{Summary of notation}\label{table1}
\begin{tabular}{| c | m{6.8cm} |}
\hline
Notation & Description \\ \hline
$K$ & Number of access ICNs \\ \hline
$P_k^{\left( n \right)}$ & Network price of ICN $k$ per unit data \\ \hline
$P_k^{\left( s \right)}$ & Storage price of access ICN or content provider $k$ per unit data \\ \hline
${P_C}$ & Transit price charge by transit ICN $C$ \\ \hline
${P_k}$ & Access ICN $k$'s price to its users per unit data \\ \hline
$P_O^{\left( c \right)}$ & Content price of content provider $O$ per unit data \\ \hline
${\rho _k}$ & Reflective coefficient of price’s influence on access ICN $k$'s user demand \\ \hline
${\sigma _k}$ & Normalized total user demands for ICN $k$ per unit data \\ \hline
${\sigma _{{k_i}}}$ & Normalized total user demands for ICN $k$ for content type $i$ per unit data \\ \hline
${\alpha _{I,{S_i}}}$ & Fraction of ICN $I$'s demand for content type $i$ that comes from cache $S$ \\ \hline
${c_{{k_i}}}$ & Caching cost of ICN $k$ for content type $i$ \\ \hline
${c_{{k_0}}}$ & Initial caching cost of ICN $k$ \\ \hline
${c_O}$ & Content provider $O$ cost \\ \hline
${\beta _k}$ & Scaling parameter between network price and storage price (greater than unity) \\ \hline
$\gamma$ & Zipf popularity law exponent  \\ \hline
$Th$ & Caching threshold index for access ICNs  \\ \hline
$Th_C$ & Caching threshold index for transit ICN $C$ \\ \hline
${q_M}\left( m \right)$ & Popularity (request rate) for content type $m$ over a set of $M$ content \\ \hline

\end{tabular}
\end{table}

As illustrated in  Fig. \ref{fig:fig2}, the interaction between different entities in the ICN model results in a conflict among the ICNs (players) when they unilaterally try to maximize their revenue. In the following section, we use a game theoretic approach to solve the joint caching and pricing strategies for each entity in the ICN network.


\section{Joint Caching and Pricing Strategy}

\subsection{Utility Function}

Each ICN can cache the content or just forward the requests to other ICNs or content provider based on the utility that it gains. The utility function for each player is defined as the utility received by providing the services for others. Each player incurs a caching cost with respect to the popularity of the content when it stores a unit of data. Therefore, as shown in Fig. \ref{fig:fig2}, the utility functions for the access ICNs $A$ and $B$, the transit ICN $C$ and the content provider $O$ can be formulated as an opportunistic function in terms of the prices, caching costs, demands and fraction of content stored and popularity.

The access ICN $A$ incurs a caching cost of ${c_{{A_i}}}$ for content type $i$ if it decides to store a unit of data. Therefore, the utility function of the access ICN $A$ is given as the average utility for all different types of content as follows:
\begin{equation} \label{eq:4}
{U_A} = \sum\limits_{i = 1}^M {{q_M}(i)} \left\{ \begin{array}{l}
{\sigma _A}{\alpha _{A,{A_i}}}\left( {{P_A} - {c_{{A_i}}}} \right) +\\
{\sigma _A}{\alpha _{A,Ou{t_i}}}\left( {{P_A} - {P_C}} \right) + \\
{\sigma _B}{\alpha _{B,{A_i}}}\left( {\left( {\frac{1}{{1 + {\beta _A}}}} \right){P_A} - {c_{{A_i}}}} \right)
\end{array} \right\},
\end{equation}
where the first term in (\ref{eq:4}) is the utility that results when ICN $A$ stores a portion of its users’ demands in its own cache. ${{q_M}(i)}$ denotes the popularity (request rate) of content type $i$ over a set of $M$ different types of content, $\sigma _A$ is the total demand that users have requested to the access ICN $A$  and $\alpha_{A,{A_i}}$ is the fraction of ICN $A$'s demand for content type $i$ that is going to be served by ICN $A$'s cache. Therefore, $\sigma _A$ multiplied by ${{q_M}(i)}$ and $\alpha_{A,{A_i}}$ is the demand of access ICN $A$ for content type $i$ that is served by transit ICN $A$'s cache and $\left( {{P_A} - {c_{{A_i}}}} \right) $ is the revenue of access ICN A by serving this portion of the requested demand.
The second term is the utility that results when ICN $A$ forwards a portion of its users’ demand to the transit ICN $C$. The third term is the utility that results when the transit ICN $C$ forwards a portion of ICN $B$'s users' demand to ICN $A$. The access ICN $A$ can control only the caching and pricing parameters ${\alpha _{A,{A_i}}}$, ${P_A}$ and $P_A^{(s)}$. Note that ${\alpha _{A,Ou{t_i}}} = 1 - {\alpha _{A,{A_i}}}$ or ${\alpha _{A,Ou{t_i}}} = {\alpha _{A,{B_i}}} + {\alpha _{A,{C_i}}} + {\alpha _{A,{O_i}}}$. The utility function of the access ICN $B$ can be defined in a similar way as follows:
\begin{equation} \label{eq:5}
{U_B} = \sum\limits_{i = 1}^M {{q_M}(i)\left\{ \begin{array}{l}
{\sigma _B}{\alpha _{B,{B_i}}}\left( {{P_B} - {c_{{B_i}}}} \right) +\\
{\sigma _B}{\alpha _{B,Ou{t_i}}}\left( {{P_B} - {P_C}} \right) + \\
{\sigma _A}{\alpha _{A,{B_i}}}\left( {\left( {\frac{1}{{1 + {\beta _B}}}} \right){P_B} - {c_{{B_i}}}} \right)
\end{array} \right\}.} 
\end{equation}

The access ICN $B$ can control only the caching and pricing parameters ${\alpha _{B,{B_i}}}$, ${P_B}$ and $P_B^{(s)}$. Note that ${\alpha _{B,Ou{t_i}}} = 1 - {\alpha _{B,{B_i}}}$ or ${\alpha _{B,Ou{t_i}}} = {\alpha _{B,{A_i}}} + {\alpha _{B,{C_i}}} + {\alpha _{B,{O_i}}}$. The transit ICN $C$ gains a profit if the access ICNs $A$ and $B$ request a content through it. Equation (\ref{eq:6}) consists of two terms. The first term is the utility that results when transit ICN $C$ stores ICN $A$ and $B$'s content in its own cache and the second term is the utility that results when it forwards ICN $A$ and $B$'s content to other ICN's or content provider. We model the ICN $C$'s utility in the following way: 

\begin{equation} \label{eq:6}
\begin{array}{*{20}{l}}
{{U_C} = \sum\limits_{X \in \{ A,B\} } {\sum\limits_{i = 1}^M {{q_M}(i){\sigma _X}} } {\alpha _{X,{C_i}}}({P_C} - {c_{{C_i}}}) + }\\
{\sum\limits_{X \in \{ A,B\} } {\sum\limits_{L \in \{ A,B,O\} ,L \ne X} {\sum\limits_{i = 1}^L {{q_M}(i){\sigma _X}{\alpha _{X,{L_i}}}} } } ({P_C} - P_L^{(s)})}
\end{array},
\end{equation}
where the transit ICN $C$ has control over caching variables ${\alpha _{A,{B_i}}},{\alpha _{A,{C_i}}},{\alpha _{A,{O_i}}},{\alpha _{B,{A_i}}},{\alpha _{B,{C_i}}},{\alpha _{B,{O_i}}}$ and ${P_C}$. Finally, if the requests are not served by any access ICNs or transit ICN, they will be forwarded to the content provider that has all the content in its cache and the costs for all types of content are identical and equal to ${c_O}$. The content provider's utility function can be expressed formally as 

\begin{equation} \label{eq:7}
\begin{array}{l}
{U_O} = \sum\limits_{i = 1}^M {{q_M}(i)\left[ {{\sigma _A}{\alpha _{A,{O_i}}} + {\sigma _B}{\alpha _{B,{O_i}}}} \right]} \left( {P_O^{\left( s \right)} - {c_O}} \right)\\
  \qquad+ \left( {{\sigma _A} + {\sigma _B}} \right)P_O^{\left( c \right)}.
\end{array}
\end{equation}
The first term in (\ref{eq:7}) results when the content provider $O$ charges transit ICN $C$ with storage price to deliver the requested content to it. The second term comes from the content price that content provider charges the users. The content provider can control pricing parameters $P_O^{(c)}$ and $P_O^{(s)}$.

Because of the competitive nature of this problem, we can present a solution in the analytical setting of a game theoretic framework. Let $G = \left[ {N,\left\{ {{S_j}} \right\},\left\{ {{U_j}\left( . \right)} \right\}} \right]$  denote the non-cooperative game among players from the set $N = \{ A,B,C,O\}$ , where ${S_j} = \left( {{P_j},{\alpha _j}} \right)$ is the set of joint caching $\left( {{\alpha _j}} \right)$ and pricing $\left( {{P_j}} \right)$ strategies and ${U_j}\left( . \right)$ is the utility function of player $j$. The strategy space of all the entities excluding the $j$th player is denoted by ${S_{ - j}}$. In the joint caching and pricing game, each player tries to maximize its own utility by solving the following optimization problem for all $j \in N$,

\begin{equation} \label{eq:8}
\mathop {{\rm{\it{max}}}}\limits_{{s_j} \in {S_j}} {\rm{   }}{U_j}\left( {{s_j},{S_{ - j}}} \right),\forall j \in N.
\end{equation}

It is necessary to characterize a set of caching and pricing strategies where all the players are satisfied with the utility they receive, given the strategy selection of other players. Such an operating point, if it exists, is called equilibrium. The notion that is most widely used for game theoretic problems is the \textit{Nash Equilibrium(\textit{NE})} \cite{Nash}. A set of pricing and caching strategies $S_j^* = \left( {P_j^*,\alpha _j^*} \right)$ constitutes a \textit{NE} if for every $j \in  N$, ${U_j}\left( {s_j^*,{S_{ - j}}} \right) \ge {U_j}\left( {{s_j},{S_{ - j}}} \right)$ for all ${s_j} \in {S_j}$. The \textit{NE} of the game is one where no player benefits by deviating from her strategy unilaterally.


\subsection{Characterization of Nash Strategies}

To find the \textit{NE} using the best response functions we need to solve the four following optimization problems for ICN $A$, ICN $B$, ICN $C$ and the content provider $O$, respectively. The maximization problem for ICN $A$ is:
\begin{equation} \label{eq:9}
\begin{array}{l}
\mathop {\it{max} }\limits_{{\alpha _{A,{A_i}}},{P_A}} {U_A} = \sum\limits_{i = 1}^M {{q_M}(i)} \left\{ \begin{array}{l}
{\sigma _A}{\alpha _{A,{A_i}}}\left( {{P_C} - {c_{{A_i}}}} \right)+\\
{\sigma _A}\left( {{P_A} - {P_C}} \right)+\\
{\sigma _B}{\alpha _{B,{A_i}}}\left( {\left( {\frac{1}{{1 + {\beta _A}}}} \right){P_A} - {c_{{A_i}}}} \right)
\end{array} \right\}\\
s.t.\begin{array}{*{20}{c}}
{}&{}&{0 \le {\alpha _{A,{A_i}}} \le 1}
\end{array},{P_A} > 0
\end{array},
\end{equation}
where the access ICN $A$ tries to maximize her utility by changing its caching (${\alpha _{A,{A_i}}}$) and pricing  (${P_A}$) strategies while other players' strategies are unknown to her. Similarly, the maximization problem for ICN $B$ is:
\begin{equation} \label{eq:10}
\begin{array}{l}
\mathop {\it{max} }\limits_{{\alpha _{B,{B_i}}},{P_B}} {U_B} = \sum\limits_{i = 1}^M {{q_M}(i)\left\{ \begin{array}{l}
{\sigma _B}{\alpha _{B,{B_i}}}\left( {{P_C} - {c_{{B_i}}}} \right)+\\
{\sigma _B}\left( {{P_B} - {P_C}} \right)+\\
{\sigma _A}{\alpha _{A,{B_i}}}\left( {\left( {\frac{1}{{1 + {\beta _B}}}} \right){P_B} - {c_{{B_i}}}} \right)
\end{array} \right\}} \\
s.t.\begin{array}{*{20}{c}}
{}&{}&{0 \le {\alpha _{B,{B_i}}} \le 1}
\end{array},{P_B} > 0
\end{array}.
\end{equation}
The transit ICN $C$ maximization problem is given as the following equation: 
\begin{equation} \label{eq:11}
\begin{array}{l}
\mathop {\it{max} }\limits_{{\alpha _{A,{C_i}}},{\alpha _{B,{C_i}}},{\alpha _{A,{B_i}}},{\alpha _{A,{O_i}}},{\alpha _{B,{A_i}}},{\alpha _{B,{O_i}}},{P_C}} {U_C} =
\sum\limits_{i = 1}^M {{q_M}(i){\sigma _A}}\\
 \times \left\{ \begin{array}{l}
{\alpha _{A,{C_i}}}\left( {{P_C} - {c_{{C_i}}}} \right) + \\
{\alpha _{A,{B_i}}}\left( {{P_C} - P_B^{\left( s \right)}} \right) + {\alpha _{A,{O_i}}}\left( {{P_C} - P_O^{\left( s \right)}} \right)
\end{array} \right\}+\\
 \sum\limits_{i = 1}^M {{q_M}(i){\sigma _B}}
 \left\{ \begin{array}{l}
{\alpha _{B,{C_i}}}\left( {{P_C} - {c_{{C_i}}}} \right) +\\
{\alpha _{B,{A_i}}}\left( {{P_C} - P_A^{\left( s \right)}} \right) + {\alpha _{B,{O_i}}}\left( {{P_C} - P_O^{\left( s \right)}} \right)
\end{array} \right\},\\
s.t.\begin{array}{*{20}{c}}
{}&{}&\begin{array}{l}
{\alpha _{A,{A_i}}} + {\alpha _{A,{B_i}}} + {\alpha _{A,{C_i}}} + {\alpha _{A,{O_i}}} = 1\\
{\alpha _{B,{B_i}}} + {\alpha _{B,{A_i}}} + {\alpha _{B,{C_i}}} + {\alpha _{B,{O_i}}} = 1\\
0 \le {\alpha _{A,{B_i}}} \le 1,0 \le {\alpha _{A,{C_i}}} \le 1,0 \le {\alpha _{A,{O_i}}} \le 1\\
0 \le {\alpha _{B,{A_i}}} \le 1,0 \le {\alpha _{B,{C_i}}} \le 1,0 \le {\alpha _{B,{O_i}}} \le 1\\
{P_C} \ge 0
\end{array}
\end{array}
\end{array}
\end{equation}
and the content provider $O$ maximizes its utility using equation (\ref{eq:12})
\begin{equation} \label{eq:12}
\begin{array}{l}
\mathop {\it{max} }\limits_{P_O^{\left( c \right)},P_O^{\left( s \right)}} {U_O} = \left( {{\sigma _A} + {\sigma _B}} \right)P_O^{\left( c \right)}\\
\begin{array}{*{20}{c}}
\end{array} \qquad + \sum\limits_{i = 1}^M {{q_M}(i)\left[ {{\sigma _A}{\alpha _{A,{O_i}}} + {\sigma _B}{\alpha _{B,{O_i}}}} \right]} \left( {P_O^{\left( s \right)} - {c_O}} \right)\\
s.t.\begin{array}{*{20}{c}}
{}&{}&{P_O^{\left( C \right)} > 0,P_O^{\left( s \right)} > {c_O}}
\end{array}
\end{array}
\end{equation}


\begin{theorem}
\label{t1}
The caching variables ${\alpha _{I,{S_i}}}$ take on values of either 0 or 1 at the equilibrium of the caching and pricing strategies game.
\end{theorem}

\begin{proof}
The solution of a maximization (minimization) problem with an objective function that has a linear relationship with the variable is the boundary point of the feasible interval. Therefore, since the relationship between utility function and caching parameters are linear to maximize the utility functions, they just take on the boundary values. Since ${\alpha _{I,{S_i}}} \in \left[ {0,1} \right]$, therefore, they can be either 0 or 1.   
\end{proof}

Given Theorem \ref{t1}, it follows that all caching variables adopt binary values. For example, according to equation (\ref{eq:9}), if ${P_C} > {c_{{A_i}}}$, the caching variable ${\alpha _{A,{A_i}}}$ should be 1 to maximize the access ICN $A$'s utility function. Whenever ${P_C} > {c_{{A_i}}}$ it means that the transit price for delivering the requested content type $i$ to users under the access ICN $A$ is smaller than the caching cost of access ICN $A$ for locally storing the requested file itself. Therefore the access ICN decides to store the content locally in its cache rather than transferring the request to transit ICN $C$. That is the reason that caching variable ${\alpha _{A,{A_i}}}$ gets the value 1. The caching strategies at equilibrium for different conditions are summarized in Table \ref{table2}.

Unlike the caching parameters which only get binary values; the pricing parameters can be continuous. Therefore, this optimization problem is a mixed integer program with multiple objective functions, and in general, the uniqueness of the NE in terms of pricing and caching strategies cannot be characterized.


\begin{table*}
\centering
\caption{Caching table for each content type $i$}\label{table2}
\begin{tabular}{| c | c | c | c | c | c | c | c | c | c |}
\hline
  &  Condition & ${\alpha _{A,{A_i}}}$ & ${\alpha _{A,{B_i}}}$& ${\alpha _{A,{C_i}}}$ & ${\alpha _{A,{O_i}}}$ & ${\alpha _{B,{B_i}}}$ & ${\alpha _{B,{A_i}}}$ & ${\alpha _{B,{C_i}}}$ & ${\alpha _{B,{O_i}}}$\\ \hline

1 & $P_C > \left({{c_{{A_i}}} \& c_{{B_i}}}  \right)$ & 1 & 0 & 0 & 0 & 1 & 0 & 0 & 0 \\ \hline

  &  $P_C > c_{A_i} \& P_C < c_{B_i} $ &  \multicolumn{8}{|c|}{ }  \\ \hline

2  & $ c_{C_i}=min\left\{ {{c_{{C_i}}},P_O^{\left( s \right)},P_A^{\left( s \right)}, P_B^{\left( s \right)}} \right\}$  & 1 & 0 & 0 & 0 & 0 & 0 & 1 & 0   \\ \hline

3 & $ P_O^{\left( s \right)}=min\left\{ {{c_{{C_i}}},P_O^{\left( s \right)},P_A^{\left( s \right)}, P_B^{\left( s \right)}} \right\}$  & 1 & 0 & 0 & 0 & 0 & 0 & 0 & 1  \\ \hline

4 & $ P_A^{\left( s \right)}=min\left\{ {{c_{{C_i}}},P_O^{\left( s \right)},P_A^{\left( s \right)}, P_B^{\left( s \right)}} \right\}$  & 1 & 0 & 0 & 0 & 0 & 1 & 0 & 0  \\ \hline

&  $P_C < c_{A_i} \& P_C > c_{B_i} $  &  \multicolumn{8}{|c|}{ } \\ \hline

5  & $ c_{C_i}=min\left\{ {{c_{{C_i}}},P_O^{\left( s \right)},P_A^{\left( s \right)}, P_B^{\left( s \right)}} \right\}$  & 0 & 0 & 1 & 0 & 1 & 0 & 0 & 0   \\ \hline

6 & $ P_O^{\left( s \right)}=min\left\{ {{c_{{C_i}}},P_O^{\left( s \right)},P_A^{\left( s \right)}, P_B^{\left( s \right)}} \right\}$  & 0 & 0 & 0 & 1 & 1 & 0 & 0 & 0  \\ \hline

7 & $ P_B^{\left( s \right)}=min\left\{ {{c_{{C_i}}},P_O^{\left( s \right)},P_A^{\left( s \right)}, P_B^{\left( s \right)}} \right\}$  & 0 & 1 & 0 & 0 & 1 & 0 & 0 & 0  \\ \hline

 & $P_C < \left({{c_{{A_i}}} \& c_{{B_i}}}  \right)$ &  \multicolumn{8}{|c|}{ }    \\ \hline

8  & $ c_{C_i}=min\left\{ {{c_{{C_i}}},P_O^{\left( s \right)},P_A^{\left( s \right)}, P_B^{\left( s \right)}} \right\}$  & 0 & 0 & 1 & 0 & 0 & 0 & 1 & 0   \\ \hline

9 & $ P_O^{\left( s \right)}=min\left\{ {{c_{{C_i}}},P_O^{\left( s \right)},P_A^{\left( s \right)}, P_B^{\left( s \right)}} \right\}$  & 0 & 0 & 0 & 1 & 0 & 0 & 0 & 1  \\ \hline

\end{tabular}
\end{table*}



\section{Generalization to $K$ Access ICNs}

We can extend the case of two access ICNs, one transit ICN and  one content provider to a generalized scenario of K access ICNs, one transit ICN and one content provider. We consider $\tilde K = \left\{ {{A_1},{A_2},...,{A_K}} \right\}$ as the set of access ICNs which are connected to transit ICN $C$ and ${\alpha}$ is the set of all caching variables. The demand function of each access ICN is defined in Definition \ref{t3}. 
 
\begin{definition}
\label{t3}
The received demands at access ICN $A_j$ is defined as

\begin{equation} \label{newdem}
\begin{array}{l}
{\sigma _{{A_j}}} = 1 - {\rho _{{A_j}}}{P_{{A_j}}} +
\frac{1}{{K - 1}}\left[ {\sum\limits_{k = 1,k \ne j}^K {{\rho _{{A_k}}}{P_{{A_k}}}} } \right]\\ \qquad- {\rho _0}P_o^{(c)},\forall j=1,...,K. 
\end{array}
\end{equation}
\end{definition}

The received demand of content type $i$ by access ICN $A_j$ can be shown as ${\sigma _{{A_j}_i}} = {\sigma _{{A_j}}}{q_M}\left( i \right)$. Following the previous section, the maximization problem of each access ICN ${A_j} \in K$ can be defined as follows:

\begin{equation} \label{access}
\begin{array}{l}
\mathop {\it{max} }\limits_{{\alpha _{{A_j},{A_{{j_i}}}}},{P_{{A_j}}}} {U_{{A_j}}} = \sum\limits_{i = 1}^M {{q_M}(i)} \left\{ \begin{array}{l}
{\sigma _{{A_j}}}{\alpha _{_{{{\left( {{A_j},{A_j}} \right)}_i}}}}\left( {{P_C} - {c_{{A_{{j_i}}}}}} \right)
  \\+{\sigma _{{A_j}}}\left( {{P_{{A_j}}} - {P_C}} \right)
 +\\ \left( {\left( {\frac{1}{{1 + {\beta _{{A_j}}}}}} \right){P_{{A_j}}} - {c_{{A_{{j_i}}}}}} \right)\times\\
\sum\limits_{k \in \tilde K,k \ne A_j} {{\sigma _k}{\alpha _{{{\left( {k,{A_j}} \right)}_i}}}} 
\end{array} \right\}\\
s.t.\begin{array}{*{20}{c}}
{}&{}&{0 \le {\rm{ }}{\alpha _{{A_j},{A_{{j_i}}}}} \le 1}
\end{array},{P_{{A_j}}} > 0
\end{array}
\end{equation}
The transit ICN $C$ maximization problem is given by the following equation:

%
%
%

\begin{equation} \label{transit}
\begin{array}{*{20}{l}}
{\mathop {\it{max} }\limits_{\alpha ,{P_C}} {U_C} = \sum\limits_{k \in \tilde K} {\sum\limits_{i = 1}^M {{q_M}(i){\sigma _k}} } {\alpha _{k,{C_i}}}({P_C} - {c_{{C_i}}})  }\\
\qquad +{\sum\limits_{k \in \tilde K} {\sum\limits_{L \in \tilde K \cup \left\{ O \right\},L \ne k} {\sum\limits_{i = 1}^M {{q_M}(i){\sigma _k}{\alpha _{k,{L_i}}}} } } ({P_C} - P_L^{(s)})}\\
\begin{array}{l}
\\
s.t.\begin{array}{*{20}{c}}
{}&\begin{array}{l}
{\alpha _{{{\left( {{A_j},{A_j}} \right)}_i}}} + \sum\limits_{\substack{L \in \tilde K \cup \left\{ {C,O} \right\} \\ L \ne {A_j}}} {{\alpha _{{{\left( {{A_j},L} \right)}_i}}}}  = 1\\
0 \le \alpha  \le 1\\
{P_C} \ge 0
\end{array}
\end{array}
\end{array}
\end{array}
\end{equation}
and the content provider $O$ maximizes its utility using equation (\ref{content}).

\begin{equation} \label{content}
\begin{array}{*{20}{l}}
{\mathop {\it{max} }\limits_{P_O^{\left( c \right)},P_O^{\left( s \right)}} {U_O} = \sum\limits_{i = 1}^M {{q_M}(i)\left[ {\sum\limits_{k \in \tilde K} {{\sigma _k}} {\alpha _{k,{O_i}}}} \right]} \left( {P_O^{\left( s \right)} - {c_O}} \right)}\\
{\begin{array}{*{20}{c}}
{}&{ \qquad+ \sum\limits_{k \in \tilde K} {{\sigma _k}} P_O^{\left( c \right)}}
\end{array}}\\
{s.t.\begin{array}{*{20}{c}}
{}&{}&{P_O^{\left( C \right)} > 0,P_O^{\left( s \right)} > {c_O}}
\end{array}}
\end{array}
\end{equation}

Theorem \ref{t1} can be extended for the generalized $K$ Access ICN case with the same reasoning and all the caching variables take binary values (i.e., all or nothing 0-1 strategies). So, the joint caching and pricing strategy game in the general form is also a mixed integer program. In the next section we will simplify the problem with the assumption of symmetric access ICNs with similar characteristics and try to give some analytical and intuitive results.


\section{Symmetric Access ICNs Scenario Analysis}

In the previous section, the general form of the caching and pricing strategies for ICNs was formulated through equations (\ref{access})-(\ref{content}) as a set of mixed integer programs. In this section, in order to analytically study the equilibrium of our proposed model, we consider the symmetric scenario, where all access ICNs have the same specifications. For the symmetric scenario, where all the access ICNs are exactly the same, we consider ${\rho _k} = \rho ,{\beta _k} = \beta ,{c_{{k_0}}} = {c_0}\forall k \in \tilde K$.

\begin{theorem}
\label{t2}
In the symmetric case, for each content type $i$, ${\alpha _{{A_k,A_j}_i}} = 0$, $\forall k \ne j$.
\end{theorem}

\begin{proof}
According to equations (\ref{access}), when access ICN $A_j$ receives a request for content type $i$ and transit price is greater than its caching cost for that particular type of content (${P_C} \ge {c_{{A_j}_i}}$), access ICN $A_j$ decides to serve the requested content itself by adopting value 1 for caching parameter ${\alpha _{{A_j,A_j}_i}}$. Since ${\alpha _{{A_j,A_j}_i}} = 1$, then all the other caching parameters for content type $i$ would be equal to zero. On the other hand, if the transit price is less than the access ICN's caching cost for content type $i$ (${P_C} < {c_{{A_j}_i}}$), the access ICN $A_j$ will forward the request to the transit ICN $C$ to be served by choosing ${\alpha _{{A_j,A_j}_i}} = 0$. When the transit ICN $C$ receives the request, it should decide to either cache the content or forward it to the content provider or other access ICNs based on the payoff that it gains according to equations (\ref{transit}). Considering the Theorem \ref{t1} and the constraint ${\alpha _{{{\left( {{A_j},{A_j}} \right)}_i}}} + \sum\limits_{L \in \tilde K \cup \{ C,O\} ,L \ne {A_j}} {{\alpha _{{{\left( {{A_j},L} \right)}_i}}}}  = 1$, one of the caching parameters should be 1 and others should be 0. If ${c_{{C_i}}}$ or $P_O^{\left( s \right)}$ are the minimums among $\left\{ {{c_{{C_i}}},P_O^{\left( s \right)},P_k^{\left( s \right)}\forall k \in \tilde K,k \ne {A_j}} \right\}$, the transit ICN $C$ (${\alpha _{{A_j,C}_i}} = 1$) or the content provider $O$ (${\alpha _{{A_j,O}_i}} = 1$) will serve the request, respectively. Now assume that one of the access ICN's storage price $P_k^{\left( s \right)}$ is the minimum. If $P_k^{\left( s \right)}$ is the minimum, it means that $P_k^{\left( s \right)} < \left( {{c_{{C_i}}}\& P_O^{\left( s \right)}} \right)$. On the other hand, $P_k^{\left( s \right)} > {c_{{K_i}}}$ in order to the access ICN $K$ accepts the request; otherwise it does not accept the content to prevent from losing payoff. Therefore, ${c_{{K_i}}} < P_k^{\left( s \right)} < {c_{{C_i}}}$. Note that ${c_{{K_i}}} = {c_{{A_j}_i}}$ in symmetric scenario. Thus $P_k^{\left( s \right)}$ should adopt a value greater than access ICN $A_j$'s caching cost and less than transit ICN $C$'s caching cost for content type $i$ (${c_{{A_j}_i}} < P_k^{\left( s \right)} < {c_{{C_i}}}$). On the other hand, we know that ${P_C} < {c_{{A_j}_i}}$, therefore ${P_C} < {c_{{A_j}_i}} < P_k^{\left( s \right)} < \left( {{c_{{C_i}}}\& P_O^{\left( s \right)}} \right)$. It means that $\left( {{P_C} - {c_{{C_i}}}} \right)$, $\left( {{P_C} - P_k^{\left( s \right)}} \right)$ and $\left( {{P_C} - P_O^{\left( s \right)}} \right)$ are negative and that is a contradiction since the transit ICN $C$ is trying to choose ${P_C}$ in a way to get at least zero payoff. Thus $P_k^{\left( s \right)}$ can never be the minimum value among the others and accordingly ${\alpha _{{A_k,A_j}_i}} = 0$, $\forall k \ne j, \forall i$,  in the symmetric scenario case. 
\end{proof}

To better understand the implication of Theorem \ref{t2}, we can refer to Table \ref{table2} that shows the caching strategies at the equilibrium for asymmetric case with two access ICNs which are not identical. By assuming two identical access ICNs that have the same characteristics, the cases 2-7 can be removed and the table will reduce to just three cases and in all of these cases ${\alpha _{{A,B}_i}} = {\alpha _{{B,A}_i}}=0$. This table can also be extended for the generalized scenario with $K$ access ICNs.  

What the above theorem reveals is that in the symmetric case, access ICNs have no motivation to serve each other's users. This is not against the philosophy of the content centric network paradigm, since in this setup the access ICNs and also the transit ICN are capable of caching the requested content. Besides, this theorem is just for the \textit{Symmetric Scenario} and in the asymmetric setup the access ICNs are able to cache requests for users of other access ICNs.

Moreover, according to Theorem \ref{t2}, when the system is symmetric, we can add the following facts to our models.
\begin{equation} \label{eq:13}
\begin{array}{l}
{\alpha _{{A_j,C}_i}} = {\alpha _{{A_k,C}_i}},\begin{array}{*{20}{c}}
{}&{\forall i}
\end{array}\\
{\alpha _{{A_j,O}_i}} = {\alpha _{{A_k,O}_i}},\begin{array}{*{20}{c}}
{}&{\forall i}
\end{array}
\end{array}
\end{equation}

So using Theorem \ref{t2} and (\ref{eq:13}), we can reformulate our maximization problem described in (\ref{access})-(\ref{content}) for symmetric case. The maximization problems for access ICN $A_j$ can be expressed in equation (\ref{eq:14}) as follows:
\begin{equation} \label{eq:14}
\begin{array}{*{20}{l}}
{\mathop {\it{max} }\limits_{{P_{{A_j}}},{\alpha _{{A_j},{A_{{j_i}}}}}} {U_{{A_j}}} = {\sigma _{{A_j}}}\sum\limits_{i = 1}^M {{q_M}(i)} \left\{ \begin{array}{l}
{\alpha _{{A_j},{A_{{j_i}}}}}\left( {{P_C} - {c_{{A_{{j_i}}}}}} \right)\\
 + \left( {{P_{{A_j}}} - {P_C}} \right)
\end{array} \right\}}\\
{s.t.\begin{array}{*{20}{c}}
{}&{{\alpha _{{A_j},{A_{{j_i}}}}} \in \left\{ {0,1} \right\},{P_{{A_j}}} > 0}
\end{array}}
\end{array}
\end{equation}
The maximization problem for transit ICN $C$ can be defined in the following equation:

\begin{equation} \label{eq:16}
\begin{array}{*{20}{l}}
{\mathop {\it{max} }\limits_{\substack{\alpha _{{A_j},{C_{_i}}} \\ {\alpha _{{A_j},{O_{_i}}}} \\ {P_C}}} {U_C} = \sum\limits_{k = 1}^K {{\sigma _{{A_k}}}} .\sum\limits_{i = 1}^M {{q_M}(i)}\times\left\{{\begin{array}{*{20}{l}}
{{\alpha _{{A_j},{C_{_i}}}}\left( {{P_C} - {c_{{C_i}}}} \right) + }\\
{{\alpha _{{A_j},{O_{_i}}}}\left( {{P_C} - P_O^{\left( s \right)}} \right)}
\end{array}} \right\}}\\

{s.t.\begin{array}{*{20}{c}}
{}&{{\alpha _{{A_j},{A_{{j_i}}}}} + {\alpha _{{A_j},{C_{_i}}}} + {\alpha _{{A_j},{O_{_i}}}} = 1,}\\
{}&{{\alpha _{{A_j},{C_{_i}}}} \in \left\{ {0,1} \right\},
{\alpha _{{A_j},{O_{_i}}}} \in \left\{ {0,1} \right\},{P_C} > 0,}
\end{array}}
\end{array}
\end{equation}
and finally, the content provider maximization problem can be reformulated as the following
\begin{equation} \label{eq:17}
\begin{array}{*{20}{l}}
{\mathop {\it{max} }\limits_{\substack{P_O^{\left( c \right)}\\P_O^{\left( s \right)}}} {U_O} = \sum\limits_{k = 1}^K {{\sigma _{{A_k}}}} \left[ {P_O^{\left( c \right)} + \sum\limits_{i = 1}^M {{q_M}(i){\alpha _{{A_j},{O_{_i}}}}} \left( {P_O^{\left( s \right)} - {c_O}} \right)} \right]}\\
{s.t.\begin{array}{*{20}{c}}
{}&{}&{P_O^{\left( C \right)} > 0,P_O^{\left( s \right)} > {c_O}}
\end{array}}
\end{array}
\end{equation}

As mentioned in Theorem \ref{t1}, the caching parameters still take on binary values. Moreover, since the content popularity probability function ${q_M}\left( i \right)$ is monotonically decreasing; according to (\ref{eq:2}), the access and transit ICNs' caching costs are monotonically increasing depending on the content type. For access ICN $A_j$, the caching parameter ${\alpha _{{A_j,A_j}_i}}$ adopts value 1 when the transit price ${P_C}$ is greater than the access ICN caching cost ${c_{{A_j}_i}}$. Hence, if ${\alpha _{{A_j,A_j}_i}}$ is 1 for content type $i$, it would also be 1 for content type $i-1$. It means that if access ICN decides to cache the content type $i$, it will cache all the other content that are more popular than it. So, there would be a caching threshold index for the number of content type that access ICN is willing to locally store. We denote the optimum caching threshold index by ${Th_A}_j$ for access ICN $A_j$. In the symmetric scenario, the caching threshold indices ${Th_A}_j$ are identical for all access ICNs, so we consider $Th$ as the caching threshold index for all the access ICNs. Since caching content types $Th+1$ to $M$ is not beneficial for access ICNs, they will forward these to the transit ICN $C$ to be served. The transit ICN $C$ should decide to either serve the content itself or forward it to somewhere else. In the symmetric case, as mentioned in Theorem \ref{t2}, in case it decides not to serve the content itself, it can forward it to the content provider $O$. According to (\ref{eq:16}), if the content provider storage price $P_O^{\left( s \right)}$ is greater than the transit ICN $C$ caching cost ${c_{{C_i}}}$, the caching parameter ${\alpha _{{A_j,C}_i}}$ adopts value 1 and the transit ICN $C$ caches the content type $i$. On the other hand, if $P_O^{\left( s \right)}$ is less than ${c_{{C_i}}}$, the caching parameter ${\alpha _{{A_j,C}_i}}$ will be 0 and ${\alpha _{{A_j,O}_i}}$ adopts value 1. In this case, the content provider $O$ will take care of the request for content type $i$. As discussed, for access ICNs, the transit ICN also can have a caching threshold index. It means that if it caches content type $i$, it would also be able to cache the content type $i-1$. So there would be a caching threshold index for the number of content type that the transit ICN $C$ is willing to locally store. We denote this caching threshold index by $Th_C$. So the transit ICN $C$ will serve the content with popularity index $Th+1$ to $Th_C$ and content with popularity index greater than $Th_C$ will be served by the content provider $O$. We can summarize these new parameters in (\ref{eq:18}) and (\ref{eq:19}).

\begin{equation} \label{eq:18}
{\alpha _{{A_j},{A_{{j_i}}}}} = \left\{ {\begin{array}{*{20}{c}}
1&{i \le Th}\\
0&{i > Th}
\end{array}} \right.,\forall j \in \left\{ {1,...,K} \right\}
\end{equation}

\begin{equation} \label{eq:19}
\begin{array}{*{20}{c}}
{{\alpha _{{A_j},{C_{{j_i}}}}} = \left\{ {\begin{array}{*{20}{c}}
1&{Th + 1 \le i \le T{h_C}}\\
0&{i \le Th}
\end{array},\forall j \in \left\{ {1,...,K} \right\}} \right.}\\
{{\alpha _{{A_j},{O_{{j_i}}}}} = \left\{ {\begin{array}{*{20}{c}}
1&{T{h_C} + 1 \le i \le M}\\
0&{i \le T{h_C}}
\end{array},\forall j \in \left\{ {1,...,K} \right\}} \right.}
\end{array}
\end{equation}
Thus, all the caching variables ${\alpha}$ will be replaced by two caching threshold indices $Th$ and $Th_C$. By (\ref{eq:18}) and (\ref{eq:19}), the problem set of (\ref{eq:14})-(\ref{eq:17}) can be rearranged using new parameters $Th$ and $Th_C$. The new maximization problem for access ICN $A_j$ is given as the following:

\begin{equation} \label{eq:20}
\begin{array}{*{20}{l}}
{\mathop {\it{max} }\limits_{{P_{{A_j}}},Th} {U_{{A_j}}} = {\sigma _{{A_j}}}\left[ {{P_{{A_j}}} - Th.{c_0} - {P_C}\sum\limits_{i = Th + 1}^{M + 1} {{q_M}(i)} } \right]}\\
{s.t.\begin{array}{*{20}{c}}
{}&{0 \le Th \le M,}
\end{array}{P_{{A_j}}} > 0}
\end{array}
\end{equation}
%
The transit ICN $C$ maximizes its utility function as the following:

\begin{equation} \label{eq:22}
\begin{array}{l}
\mathop {\it{max} }\limits_{{P_C},T{h_C}} {U_C} = \left( {K - K{\rho _0}P_O^{\left( c \right)}} \right)\left\{ \begin{array}{l}
{P_C}\sum\limits_{i = Th + 1}^{M + 1} {{q_M}(i)}
 \\- \left( {T{h_C} - Th} \right){c_{{C_0}}}
 \\- P_O^{\left( s \right)}\sum\limits_{i = T{h_C} + 1}^{M + 1} {{q_M}(i)} 
\end{array} \right\}\\
s.t.\begin{array}{*{20}{c}}
{}&{Th \le T{h_C} \le M,}
\end{array}{P_C} > 0
\end{array}
\end{equation}
And the content provider $O$ maximization problem is formulated using new parameters in (\ref{eq:23}) 

\begin{equation} \label{eq:23}
\begin{array}{l}
\mathop {\it{max} }\limits_{P_O^{\left( c \right)},P_O^{\left( s \right)}} {U_O} = \left( {K - K{\rho _0}P_O^{\left( c \right)}} \right)\left\{ \begin{array}{l}
P_O^{\left( c \right)} +\left( {P_O^{\left( s \right)} - {c_O}} \right).\\
\sum\limits_{i = T{h_C} + 1}^{M + 1} {{q_M}(i)} 
\end{array} \right\}\\
s.t.\begin{array}{*{20}{c}}
{}&{}&{P_O^{\left( C \right)} > 0,P_O^{\left( s \right)} > {c_O}}
\end{array}
\end{array}
\end{equation}
Note that in the above equations ${q_M}\left( 0 \right) = {q_M}\left( {M + 1} \right) = 0$.

The problem of joint caching and pricing strategies for the case of symmetric ICNs can be decomposed into two independent caching and pricing optimization problems. The caching problem is dealing with the parameters that affect the caching process and is stated as follows:

\begin{description}
\item[$\bullet$ Caching Problem:] 
\end{description}

\begin{equation} \label{eq:24}
\begin{array}{l}
\mathop {\it {min} } \limits_{Th} Th.{c_0} + {P_C}\sum\limits_{i = Th + 1}^{M + 1} {{q_M}(i)} \\
\begin{array}{*{20}{c}}
{s.t.}&{0 \le Th \le M}
\end{array}
\end{array}
\end{equation}

\begin{equation} \label{eq:25}
\begin{array}{l}
\mathop {\it{max} }\limits_{{P_C},T{h_C}} {P_C}\sum\limits_{i = Th + 1}^{M + 1} {{q_M}(i)}  + \left( {Th - T{h_C}} \right){c_{{C_0}}}\\
\begin{array}{*{20}{c}}
\end{array} \qquad - P_O^{\left( s \right)}\sum\limits_{i = T{h_C} + 1}^{M + 1} {{q_M}(i)} \\
s.t.\begin{array}{*{20}{c}}
{}&{Th \le T{h_C} \le M,{P_C} > 0}
\end{array}
\end{array}
\end{equation}

\begin{equation} \label{eq:26}
\begin{array}{l}
\mathop {\it{max} }\limits_{P_O^{\left( s \right)}} \left( {P_O^{\left( s \right)} - {c_O}} \right).\sum\limits_{i = T{h_C} + 1}^{M + 1} {{q_M}(i)} \\
\begin{array}{*{20}{c}}
{s.t.}&{P_O^{\left( s \right)} > }
\end{array}{c_O}
\end{array}
\end{equation}
The outcome of this problem is a 4-tuples $\left( {T{h^*},T{h_C}^*,{P_C}^*,P{{_O^{\left( s \right)}}^*}} \right)$. The pricing problem is defined in (\ref{eq:27}) and (\ref{eq:29}) by substituting the 4-tuple resulting from the caching problem.

\begin{description}
\item[$\bullet$ Pricing Problem:] 
\end{description}

\begin{equation} \label{eq:27}
\begin{array}{l}
\mathop {\it{max} }\limits_{{P_A}_j} {{U_A}}_j = {\sigma _A}_j\left[ {{P_A}_j - T{h^*}.{c_0} - {P_C}^*\sum\limits_{i = T{h^*} + 1}^{M + 1} {{q_M}(i)} } \right]\\
s.t.\begin{array}{*{20}{c}}
{}&{{P_A}_j > 0}
\end{array}
\end{array}
\end{equation}
\begin{equation} \label{eq:29}
\begin{array}{l}
\mathop {\it{max} }\limits_{P_O^{\left( c \right)}} {U_O} = \left( {K - K{\rho _0}P_O^{\left( c \right)}} \right)\left\{ \begin{array}{l}
P_O^{\left( c \right)} +\sum\limits_{i = T{h_C}^* + 1}^{M + 1} {{q_M}(i)}\\
\times \left( {P{{_O^{\left( s \right)}}^*} - {c_O}} \right)
\end{array} \right\}\\
s.t.\begin{array}{*{20}{c}}
{}&{P_O^{\left( c \right)} > 0}
\end{array}
\end{array}
\end{equation}
The K+1-tuple $\left( {{P_A}_1^*,...,{P_A}_K^*,P{{_O^{\left( c \right)}}^*}} \right)$ is the outcome of the pricing problem. The \textit{NE} of the joint caching and pricing problem is $\left( {T{h^*},T{h_C}^*, {P_C}^*,P{{_O^{\left( s \right)}}^*},{P_A}_1^*,...,{P_A}_K^*,P{{_O^{\left( c \right)}}^*}} \right)$.

\begin{theorem}
\label{t3}
The caching problem introduced above is a two player matrix game between transit ICN $C$ and content provider $O$. 
\end{theorem}

\begin{proof}
According to (\ref{access}), when the transit price for ICN $C$ is greater than access ICN's caching cost for content type $i$ $\left( {{P_C} \ge {c_{{A_j}_i}}} \right)$, the access  ICN caches all the content which are more popular than content type $i$. Therefore, when ${c_{{A_j}_i}} \le {P_C} < {c_{{A_j}_{i + 1}}}$, the optimum caching threshold index chosen by access ICN will be $T{h^*} = i$. On the other hand, since the utility function of the transit ICN $C$ has a linear relationship with the transit price $P_C$ , the transit ICN will choose the maximum value possible that is ${c_{{A_j}_{i + 1}}} - \varepsilon$ ($\varepsilon$ is a very small value). Thus, the transit ICN $C$ also can adopt its actions from a discrete set. There is a caching threshold index $Th$ corresponding with each transit price chosen by ICN $C$. It shows that the transit ICN is the leader in its relationship with the access ICNs and its action is $\left( {{P_C},Th} \right)$ from a set of $M+1$ feasible choices. The relationship between the transit ICN $C$ and the content provider is also a leader follower game. Depending on the storage price $P_O^{\left( s \right)}$, the transit ICN $C$ might forward some part of demands to the content provider $O$ to be served. If the access ICN decides to cache the content more popular than content type $Th$, the rest of the content should be forwarded to transit ICN. Therefore, the content type with index $Th+1$ to $M$ is going to be served in either the transit ICN or the content provider. When ${c_{{C_j}}} \le P_O^{\left( s \right)} < {c_{{C_{j + 1}}}}$, the optimum caching threshold index chosen by the transit ICN $C$ will be $T{h_C}^* = j$. Since, the utility of the content provider $O$ has a linear relationship with the storage price $P_O^{\left( s \right)}$; the content provider will pick the maximum value possible for the storage price that is ${c_{{C_{j + 1}}}} - \varepsilon$. Thus, for every content provider storage price, there is a corresponding caching threshold index chosen by the transit ICN $C$. Since both the transit ICN $C$ and the content provider $O$ have a limited set of discrete actions, the problem introduced in (\ref{eq:24})-(\ref{eq:26}) is a matrix game between the transit ICN $C$ and the content provider $O$ when the transit ICN action is the transit price ${P_C}$ and content provider action is the storage price $P_O^{\left( s \right)}$. The access ICNs cannot change the results and they just follow the transit ICN and their actions. 
\end{proof}
By Theorem \ref{t3}, we can discard (\ref{eq:24}) and solve equations (\ref{eq:25}) and (\ref{eq:26}) jointly to find the integer thresholds $Th$ and $Th_C$. Note that $P_C$ and $P_O^{\left( s \right)}$ are functions of $Th$ and $Th_C$, respectively as follows:

\begin{equation} \label{eq:30}
{P_C}\left( {Th} \right) = \left\{ \begin{array}{l}
{c_{{A_j}_{Th + 1}}} - \varepsilon \begin{array}{*{20}{c}}
{}&{0 \le Th \le M - 1}
\end{array}\\
{c_{{A_j}_{Th}}} + \varepsilon \begin{array}{*{20}{c}}
{}&{Th = M}
\end{array}
\end{array} \right.
\end{equation}

\begin{equation} \label{eq:31}
P_O^{\left( s \right)}\left( {T{h_C}} \right) = \left\{ \begin{array}{l}
{c_{{C_{Th + 1}}}} - \varepsilon \begin{array}{*{20}{c}}
{}&{0 \le T{h_C} \le M - 1}
\end{array}\\
{c_{{C_{Th}}}} + \varepsilon \begin{array}{*{20}{c}}
{}&{T{h_C} = M}
\end{array}
\end{array} \right.
\end{equation}


\begin{proposition}
\label{p2}
$f\left( {Th} \right) = {P_C}\left( {Th} \right)\sum\limits_{i = Th + 1}^{M + 1} {{q_M}(i)}  + Th{c_{{C_0}}}$ and $g\left( {T{h_C}} \right) = \left( {P_O^{\left( s \right)}\left( {T{h_C}} \right) - {c_O}} \right).$ $\sum\limits_{i = T{h_C} + 1}^{M + 1} {{q_M}(i)}$ are concave sequences and have a unique maximum.
\end{proposition}
The proof of the above proposition can be found in \textit{Appendix A}.



\begin{theorem}
\label{t4}
The symmetric joint caching and pricing game has a unique \textit{NE}.
\end{theorem}

\begin{proof}
The solution of the caching and pricing problems in (\ref{eq:24})-(\ref{eq:29}) is the \textit{NE}. On the one hand, the caching problem set is like a leader-follower game and the transit ICN $C$ (leader) maximizes (\ref{eq:25}) and the content provider $O$ (follower) maximizes (\ref{eq:26}). Since both of them are concave sequences based on Proposition \ref{p2}, they have only one maximum in their feasible sets. By (\ref{eq:25}), $T{h^*}$ can be defined as the unique optimum value for the access ICN caching threshold index and by (\ref{eq:26}), $T{h_C}^*$ can be defined as the optimum caching threshold index of the transit ICN $C$ that should always be greater or equal than $T{h^*}$. Assume that $T{h_{{C_{\max }}}}$ is the value that maximizes (\ref{eq:26}). If $T{h_C}^* < T{h_{{C_{\max }}}}$ then $T{h_C}^* = T{h_{{C_{\max }}}}$ and if $T{h_C}^* \ge T{h_{{C_{\max }}}}$ then $T{h_C}^* = T{h^*}$. By finding the first set of parameters and substituting them in (\ref{eq:27}) and (\ref{eq:29}), we can find the second set of parameters. Since these are concave quadratic functions the problem has the following unique solution
\begin{equation} \label{eq:32}
\begin{array}{l}
P{{_O^{\left( c \right)}}^*} = \it{max} \left\{ {0,\frac{{1 - {\rho _0}\left( {P{{_O^{\left( s \right)}}^*} - {c_O}} \right)\sum\limits_{i = T{h_C}^* + 1}^{M + 1} {{q_M}(i)} }}{{2{\rho _0}}}} \right\},\\
{P_A}^ *  = {P_B}^ *  = \frac{1}{\rho }\left[ \begin{array}{l}
\rho \left( {T{h^*}.{c_0} + {P_C}^*\sum\limits_{i = T{h^*} + 1}^{M + 1} {{q_M}(i)} } \right)\\+
1 - {\rho _0}P{{_O^{\left( c \right)}}^*}
\end{array} \right].
\end{array}
\end{equation}
Note that ${P_A}_j^ *$'s are always greater than zero. Hence, as we have unique set of results for $T{h^*}$ and $T{h_C}^*$, so ${P_A}_j^ *$'s and $P{{_O^{\left( c \right)}}^*}$ are also unique. Therefore the \textit{NE} exists and is unique.
\end{proof}

As the analytical results show, the \textit{NE} for the symmetric case is independent of number of access ICNs. So, for the numerical results section, we consider the scenario with only two access ICNs.


\section{Numerical Results}

We consider the interaction among two symmetric access ICNs, one transit ICN and a content provider who are competing to maximize their utilities. In this scenario, the reflective coefficients of price's influence on users' demands ${\rho _A}$, ${\rho _B}$ and ${\rho _0}$ are set identically to 0.1. These parameters model the sensitivity of the demands to an increase in the prices. The scaling parameter between network price and storage price is set to ${\beta _A} = {\beta _B} = 10$, i.e., the storage price is an order of magnitude less than the network price. There are $M=100$ different types of content in this network which are randomly requested by users according to a generalized Zipf distribution. Fig. \ref{fig:accessC2} shows how the access ICN caching cost varies for different types of content and different values of Zipf’s factor $\gamma$ when initial caching cost $c_0=1$. As Zipf's factor $\gamma$ is increasing, the distribution of content requested is getting skewed and according to (\ref{eq:2}), the caching costs of more popular content decrease while the caching costs of less popular content increase. Note that the transit ICN $C$ possibly has (on average) access to cheaper caches. We denote the ratio of transit ICN caching cost to access ICN caching cost by $R$. We assume that the content provider cost ($c_O$) for all types of content is identical.

\begin{figure}
\hspace{-0.7cm}
  \includegraphics[width=10cm,height=7.3cm]{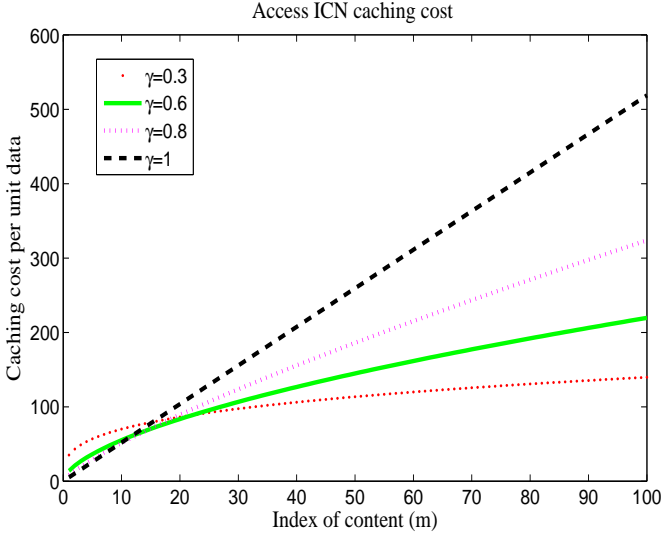}
  \caption{Access ICN caching costs vs index of content for different Zipf’s factor $\gamma$.}
  \label{fig:accessC2}
\end{figure}

\begin{figure}
\hspace{-0.7cm}
  \includegraphics[width=10cm,height=7.3cm]{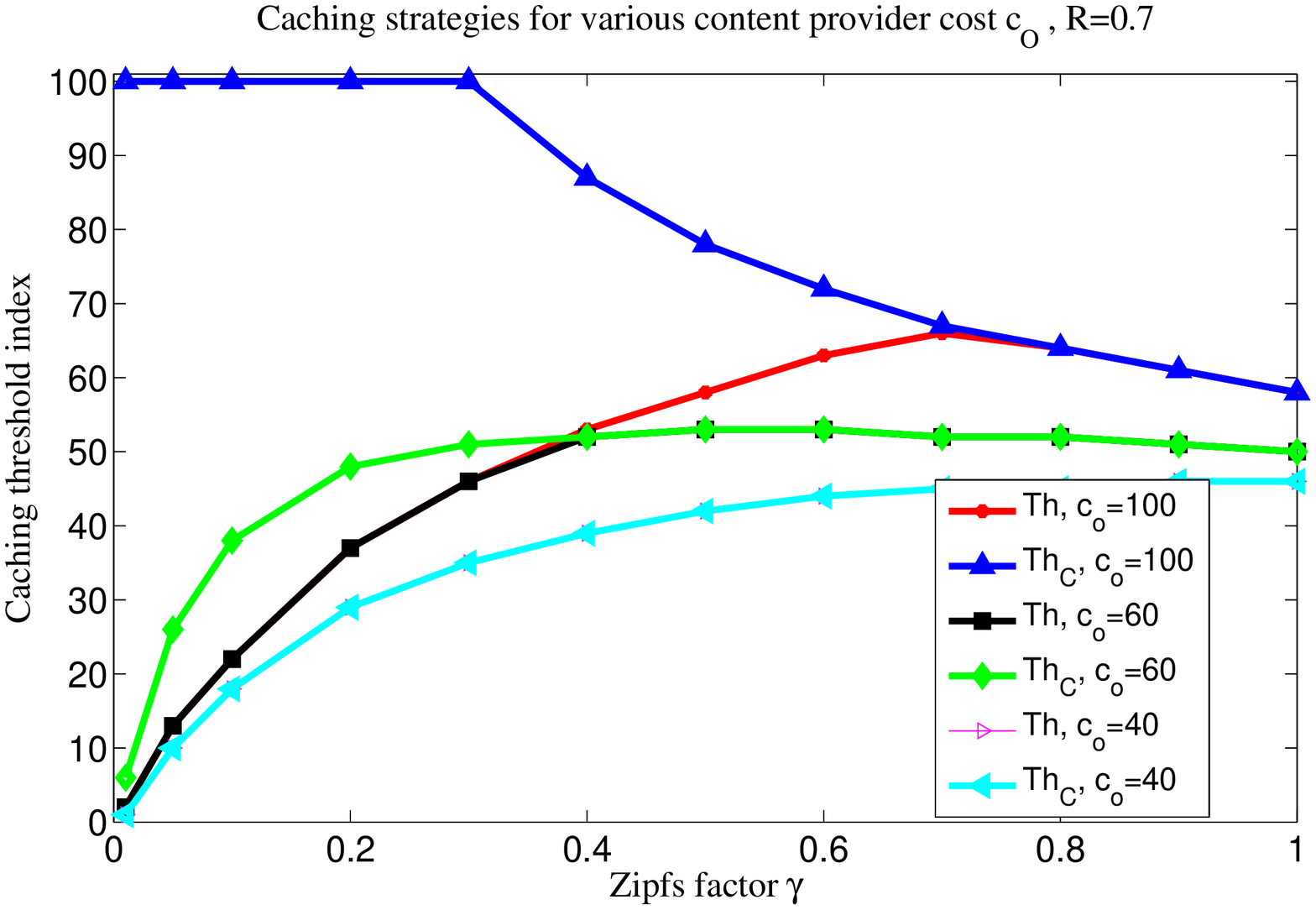}
  \caption{Access ICN/transit ICN Caching threshold vs Zipf’s factor $\gamma$ for various content provider cost $c_O$ when $R=0.7$. }
  \label{fig:threshold7}
\end{figure}

\begin{figure}
\hspace{-0.7cm}
  \includegraphics[width=10cm,height=7.3cm]{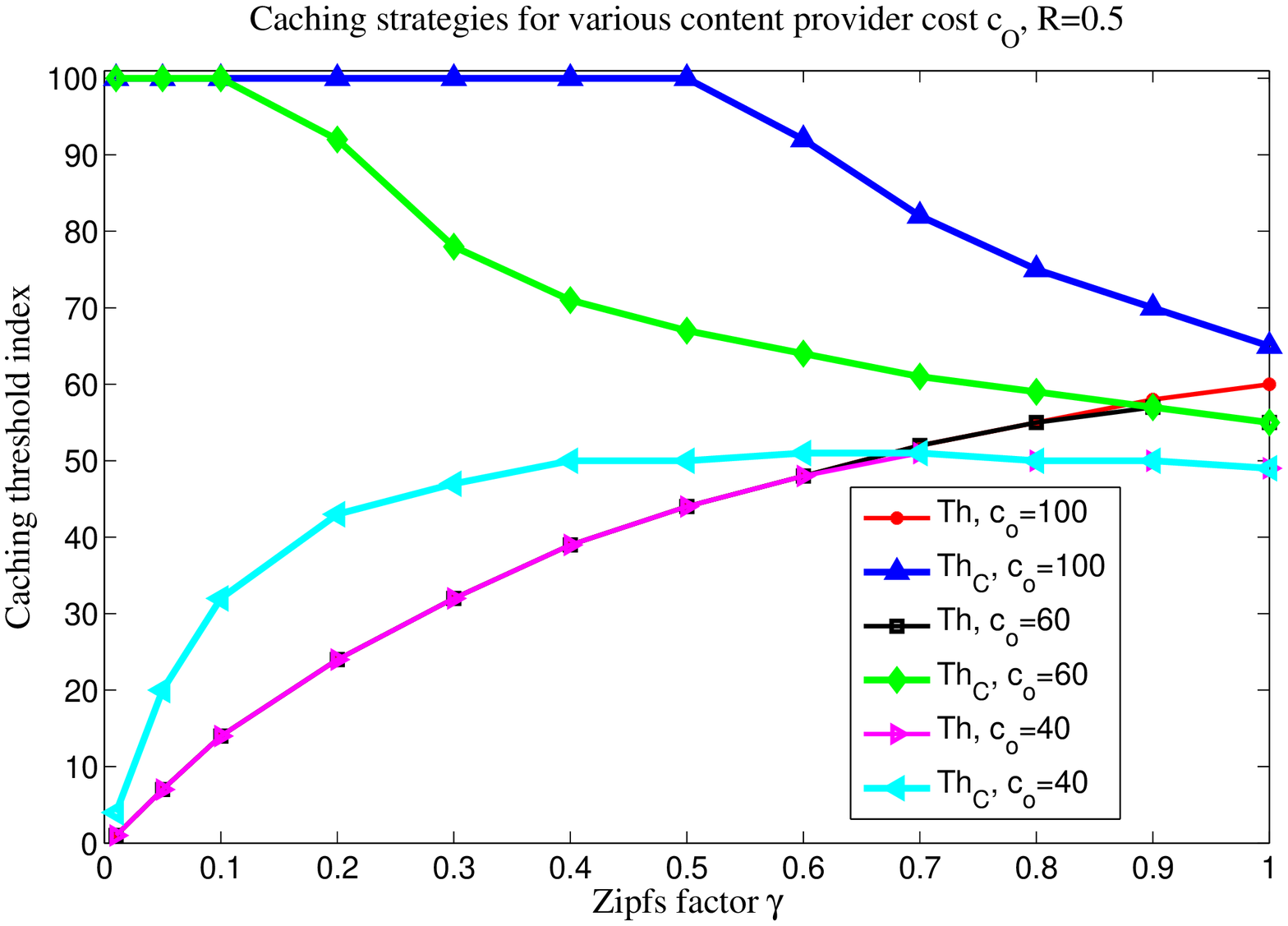}
  \caption{Access ICN/transit ICN Caching threshold vs Zipf’s factor $\gamma$ for various content provider cost $c_O$ when $R=0.5$. }
  \label{fig:threshold5}
\end{figure}

\begin{figure}[t]
\hspace{-0.7cm}
  \includegraphics[width=10cm,height=7.3cm]{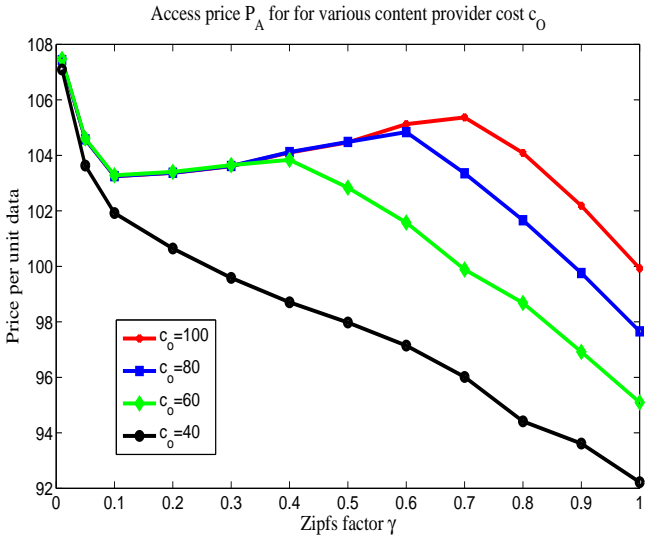}
  \caption{Access price $P_A$ vs Zipf’s factor $\gamma$ for various content provider cost $c_O$ and $R=0.7$.}
  \label{fig:access}
\end{figure}

The caching strategies $T{h^*}$ and $T{h_C}^*$ at the equilibrium are shown in Fig. \ref{fig:threshold7}, where Zipf’s factor $\gamma$ is varying in the interval of 0.01 to 1. In this scenario, the caching cost of transit ICN C is set to $70\%$ of the caching cost of access ICNs for every type of content $i$, i.e., $R=0.7$. The results are compared for $c_O=40,60,100$. As seen, for small amount of $\gamma$ the content are less skewed and the caching cost for them is very similar. Since the caching cost of transit ICN is less than access ICNs' and content provider's, it decides to cache most of the content. But as  $\gamma$ increases, some of the content is getting relatively more popular. In this case, the access ICNs prefer to cache the more popular content locally and smaller amount of content will be cached by transit ICN and content provider. For the higher $\gamma$, the caching cost of access ICNs and transit ICN is getting higher and they do not have an incentive to locally store them. Therefore, at this point, the content provider starts to serve more content than before and the transit ICN does not cache content and just forwards requests to the content provider. For higher content provider cost $c_O=100$, since the cost is so high the content provider just serves the less popular requested content but as the cost $c_O$ increases, it starts to serve more content. For the case which $c_O$ is relatively small ($c_O=40$), the transit ICN does not have the incentive to spend its resources for caching the requested content and just forwards all the requests for content from access ICNs to the content provider. The same scenario with $R=0.5$ is examined in Fig. \ref{fig:threshold5}. In this case, the transit ICN caching cost is half of the access ICN caching cost for each type of content. So, the transit ICN caches more requested content in its cache. 

Fig. \ref{fig:access} shows the access ICN price $P_A$ as function of $\gamma$ for different content provider costs $c_O$. $P_A$ decreases as the $\gamma$ gets higher when the content provider cost $c_O$ is relatively low. The reason is that as $\gamma$ increases, the caching costs of more popular content at access ICNs caches are getting lower. Therefore, the access ICNs need to decrease their price $P_A$ in order to compete with other access ICNs. But as the relative content provider cost $c_O$ increases, the access price is getting higher since both the access ICN and the transit ICN have a greater incentive to locally cache the content. However, the price for greater value of $c_O$ displays a bimodal behavior as a function of $\gamma$. According to Fig. \ref{fig:accessC2}, for relatively small value of $\gamma$, the different types of content have similar popularity and the access and transit ICN should incur more or less similar caching costs for each type of content. But as $\gamma$ increases, some of the content is getting more popular and the cost of caching them at the access ICN is decreasing. Therefore, the access ICN decreases its access price in order to induce increased demand from the users (see relationship between $P_A$ and demand in equation (\ref{eq:3})). On the other hand, when $\gamma$ keeps increasing, the transit price and the content provider storage price increase. Thus, the access ICN needs to slightly increase its price to compensate the increase in the transit price. After the slight increase, as $\gamma$ increases further, the access price $P_A$ decreases again. This is because, as seen from Fig. 4, the caching strategies $T{h^*}$ and $T{h_C}^*$ are the same for large $\gamma$. At this point, the access ICNs decide to cache the content that are more popular to get higher payoff. As a result, the access ICN decreases the price to further incentivize greater user demand for popular content.

The transit price $P_C$ and the content provider storage price $P_O^{\left( s \right)}$ are shown in Fig. \ref{fig:transit}. As $\gamma$ increases, the caching costs of less popular content that are going to be cached in transit ICN or content provider caches are increasing. Therefore, both the transit ICN and content provider should increase their prices as shown in the figure.

\begin{figure}
\hspace{-0.7cm}
  \includegraphics[width=10cm,height=7.3cm]{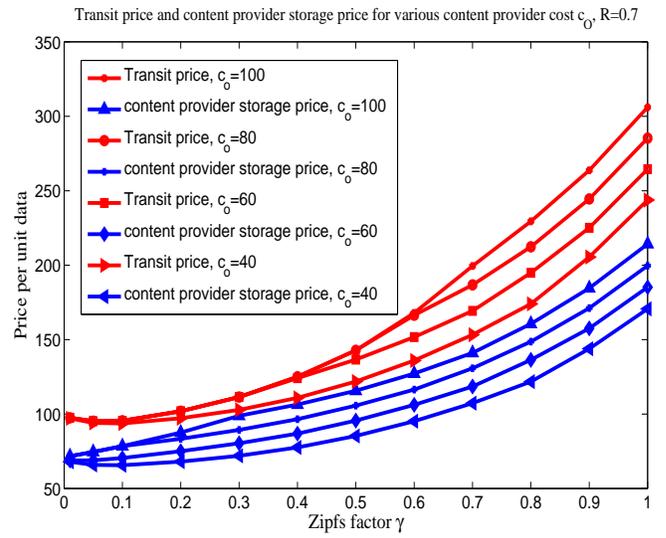}
  \caption{Transit price / content provider storage price vs Zipf’s factor $\gamma$ for various content provider cost $c_O$ and $R=0.7$.}
  \label{fig:transit}
\end{figure}



\section{Conclusion}

In this paper, we developed an analytical framework for distribution of popular content in an Information Centric Network (ICN) that comprises of Access ICNs, a Transit ICN and a Content Provider. By modeling the interaction of the above entities using game theory and under the assumption that the caching cost of the access ICNs and transit ICNs is inversely proportional to popularity, which follows a generalized Zipf distribution, we first showed that at the NE, the caching strategies turn out to be 0-1 (all or nothing). Further, for the case of symmetric Access ICNs, we showed that a unique \textit{NE} exists and the caching policy (0 or 1) is determined by a threshold on the popularity of the content (reflected by the Zipf probability metric), i.e., all content more popular than the threshold value is cached. Moreover, we showed that the resulting threshold indices and prices can be obtained by a decomposition of the joint caching and pricing problem into two independent caching only and pricing only problems. Finally, using numerical results we showed that as the Zipf's factor $\gamma$ increases and the relative popularity of the content gets skewed, the access ICN just caches the more popular content and the content provider serves only requests for less popular content while the transit ICN just forwards the demands to the content provider without locally caching any content itself. The insights obtained from the analysis here warrants further investigation into the case of asymmetric access ICNs. 
In this paper, we discussed a hierarchical scenario with $K$ access ICNs, one transit ICN and one content provider. In the special case that the transit ICNs are not "interconnected' to each other, the model in this paper is readily extendable for multiple transit ICNs. However, a generalized scenario with either multiple "interconnected" transit ICNs or multiple content providers is a direction for future work. Extension of the above model to scenarios where the content popularity may need to be learned is also of interest.


%

\appendices
\section{Proof of Proposition \ref{p2}}

Two sequences $f\left( {Th} \right)$ and $g\left( {Th_C} \right)$ defined in Proposition \ref{p2} are concave sequences. 

A sequence $S$ is strictly concave if the below inequality holds for every $n$.

\begin{center}
$S\left( {n + 1} \right) + S\left( {n - 1} \right) - 2S\left( n \right) < 0$
\end{center}

As stated before, ${q_M}\left( i \right) = \frac{1}{{\left( {\sum\limits_{j = 1}^M {\frac{1}{{{j^\gamma }}}} } \right){i^\gamma }}}$ for $ = 1,...,M$ and ${q_M}\left( i \right) = 0$ for $i>M$. Using (\ref{eq:30}) for $Th = 1,...,M - 1$, we have

\begin{center}
${P_C}\left( {Th} \right) = {c_0}{\left( {Th + 1} \right)^\gamma }\sum\limits_{i = 1}^M {{{\left( {\frac{1}{i}} \right)}^{^\gamma }}}  - \varepsilon$
\end{center}

Since $\varepsilon$ is so small and tends to zero, it can be considered as zero. Therefore, for $Th = 1,...,M - 1$, the sequence $f$ is defined as follows

\begin{center}
$\begin{array}{l}
f\left( {Th} \right) = \left[ {{c_0}{{\left( {Th + 1} \right)}^\gamma }\sum\limits_{i = 1}^M {{{\left( {\frac{1}{i}} \right)}^{^\gamma }}} } \right]\sum\limits_{i = Th + 1}^{M + 1} {{q_M}(i)}  + Th{c_{{C_0}}}\\
 \Rightarrow f\left( {Th} \right) = \left[ {{c_0}{{\left( {Th + 1} \right)}^\gamma }} \right]\sum\limits_{i = Th + 1}^{M + 1} {{{\left( {\frac{1}{i}} \right)}^{^\gamma }}}  + Th{c_{{C_0}}}
\end{array}$
\end{center}
To prove concavity, we need to have

\begin{center}
$\begin{array}{l}
f\left( {Th + 1} \right) + f\left( {Th - 1} \right) - 2f\left( {Th} \right) < 0 \Rightarrow \\
\left[ {{{\left( {Th + 2} \right)}^\gamma } - 2{{\left( {Th + 1} \right)}^\gamma } + T{h^\gamma }} \right]{\sum\limits_{i = Th + 2}^{M + 1} {\left( {\frac{1}{i}} \right)} ^\gamma }\\ < 1 - {\left( {\frac{{Th}}{{Th + 1}}} \right)^\gamma }.
\end{array}$
\end{center}
If this inequality holds for all $M$ and $Th = 1,...,M - 1$, then the sequence $f$ is concave. We introduce  ${\varphi _\gamma }\left( {Th} \right) = \left[ {{{\left( {Th + 2} \right)}^\gamma } - 2{{\left( {Th + 1} \right)}^\gamma } - T{h^\gamma }} \right]$. Two functions ${\varphi _0}\left( {Th} \right)$ and ${\varphi _1}\left( {Th} \right)$ are zero for all $Th$ that satisfy the above inequality. It means that for $\gamma=0$ and $\gamma=1$, $f$ is a concave sequence. Moreover, ${\varphi _\gamma }\left( {Th} \right)$ can be written as 
\begin{center}
$\begin{array}{l}
{\varphi _\gamma }\left( {Th} \right) = \left[ {{{\left( {Th + 2} \right)}^\gamma } - {{\left( {Th + 1} \right)}^\gamma }} \right] - \left[ {{{\left( {Th + 1} \right)}^\gamma } + T{h^\gamma }} \right]\\
 \Rightarrow {\varphi _\gamma }\left( {Th} \right) = {\left( {\frac{{Th}}{{Th + 1}}} \right)^\gamma } + {\left( {\frac{{Th + 2}}{{Th + 1}}} \right)^\gamma } - 2.
\end{array}$
\end{center}
Since ${\left( {\frac{{Th}}{{Th + 1}}} \right)^\gamma } + {\left( {\frac{{Th + 2}}{{Th + 1}}} \right)^\gamma } < 2$ for $0<\gamma<1$, ${\varphi _\gamma }\left( {Th} \right)$ is negative for all $Th = 1,...,M - 1$ and $M$ which cause that above inequality be satisfied. Therefore, $f$ is a concave sequence for all $M$ and all $0 \le \gamma  \le 1$ .

Using (\ref{eq:31}) for $Th_C = 1,...,M - 1$, we have

\begin{center}
$P_O^{\left( s \right)}\left( {T{h_C}} \right) = {c_{{C_0}}}{\left( {T{h_C} + 1} \right)^\gamma }\sum\limits_{i = 1}^M {{{\left( {\frac{1}{i}} \right)}^{^\gamma }}}  - \varepsilon$
\end{center}
Since $\varepsilon$ is so small and tends to zero, it can be considered as zero. Therefore, for $Th_C = 1,...,M - 1$, the sequence $g$ can be defined as

\begin{center}
$g = {c_{{C_0}}}{\left( {T{h_C} + 1} \right)^\gamma }\sum\limits_{i = T{h_C} + 1}^M {{{\left( {\frac{1}{i}} \right)}^{^\gamma }}}  - {c_O}\sum\limits_{i = T{h_C} + 1}^{M + 1} {{q_M}(i)}$ 
\end{center}
On the other hand, we can show that if two sequences are concave the sum of them is also concave. Assume that $g = \Psi  + \Delta$. If $\Psi$ and $\Delta$ are concave, we have

\begin{center}
$\Psi \left( {n + 1} \right) + \Psi \left( {n - 1} \right) - 2\Psi \left( n \right) < 0$
\end{center}
and
\begin{center}
$\Delta \left( {n + 1} \right) + \Delta \left( {n - 1} \right) - 2\Delta \left( n \right) < 0$.
\end{center}
Therefore, 
\begin{center}
$\begin{array}{l}
\Psi \left( {n + 1} \right) + \Delta \left( {n + 1} \right) + \Psi \left( {n - 1} \right) + \Delta \left( {n - 1} \right)\\
 - 2\left[ {\Psi \left( n \right) + \Delta \left( n \right)} \right] < 0
\end{array}$
\end{center}
That means 

\begin{center}
$g\left( {n + 1} \right) + g\left( {n - 1} \right) - 2g\left( n \right) < 0$
\end{center}

We already showed that the first term of sequence $g$ is concave, so we just need to show that the second term   $\Delta \left( {T{h_C}} \right) =  - {c_O}\sum\limits_{i = T{h_C} + 1}^{M + 1} {{q_M}(i)} $ is also concave. Then, we have to show that the following inequality holds for every $Th_C = 1,...,M - 1$

\begin{center}
$\Delta \left( {T{h_C} + 1} \right) + \Delta \left( {T{h_C} - 1} \right) - 2\Delta \left( {T{h_C}} \right) < 0$
\end{center}
By doing some algebra, we will get ${q_M}\left( {T{h_C}} \right) > {q_M}\left( {T{h_C} + 1} \right)$ that holds for $Th_C = 1,...,M - 1$, since ${q_M}\left( i \right)$ is a strictly decreasing function. That completes the proof. \hfill $ \square $


%
%
%

\ifCLASSOPTIONcaptionsoff
  \newpage
\fi

\begin{IEEEbiography}{Mohammad Hajimirsadeghi}
(S’15) received the B.S. degree
in electrical engineering (with honors, first rank)
from Shahed University, Tehran, Iran, in 2008 and the
M.S. degree in telecommunications from Sharif University
of Technology (SUT), Tehran, Iran, in 2011.
Currently, he is pursuing a Ph.D. degree in
electrical and computer engineering at WINLAB at Rutgers-The
State University of New Jersey, Piscataway, NJ under the guidance of Dr. Narayan Mandayam.
His research interests are in
the general area of communications and networking, including wireless resource allocation, Network economic and network security. Before coming to the United
States, he was a Communication Engineer at Iran
Telecommunication Research Center (ITRC), Tehran, Iran.
\end{IEEEbiography}

\begin{IEEEbiography}
{Narayan. B. Mandayam}
(S’89-M’94- SM’99-F’09) received the B.Tech (Hons.) degree in 1989 from the Indian Institute of Technology, Kharagpur, and the M.S. and Ph.D. degrees in 1991 and 1994 from Rice University, all in electrical engineering. Since 1994 he has been at Rutgers University where he is currently a Distinguished Professor and Chair of the Electrical and Computer Engineering department. He
also serves as Associate Director at WINLAB. He was a visiting faculty fellow in the Department of
Electrical Engineering, Princeton University, in 2002 and a visiting faculty at the Indian Institute of
Science, Bangalore, India in 2003. Using constructs from game theory, communications and networking,
his work has focused on system modeling, information processing as well as resource management for
enabling cognitive wireless technologies to support various applications. He has been working recently on
the use of prospect theory in understanding the psychophysics of data pricing for wireless networks as
well as the smart grid. His recent interests also include privacy in IoT as well as modeling and analysis of
trustworthy knowledge creation on the Internet.

Dr. Mandayam is a co-recipient of the 2015 IEEE Communications Society Advances in Communications
Award for his seminal work on power control and pricing, the 2014 IEEE Donald G. Fink Award for his
IEEE Proceedings paper titled “Frontiers of Wireless and Mobile Communications” and the 2009 Fred W.
Ellersick Prize from the IEEE Communications Society for his work on dynamic spectrum access models
and spectrum policy. He is also a recipient of the Peter D. Cherasia Faculty Scholar Award from Rutgers
University (2010), the National Science Foundation CAREER Award (1998) and the Institute Silver Medal
from the Indian Institute of Technology (1989). He is a coauthor of the books: Principles of Cognitive
Radio (Cambridge University Press, 2012) and Wireless Networks: Multiuser Detection in Cross-Layer
Design (Springer, 2004). He has served as an Editor for the journals IEEE Communication Letters and
IEEE Transactions on Wireless Communications. He has also served as a guest editor of the IEEE JSAC
Special Issues on Adaptive, Spectrum Agile and Cognitive Radio Networks (2007) and Game Theory in
Communication Systems (2008). He is a Fellow and Distinguished Lecturer of the IEEE.
\end{IEEEbiography}

\vspace{60 mm}


\begin{IEEEbiography}
{Alex Reznik}
[M’94, SM’10] is currently a Solution Architect on Hewlett Packard Enterprise’s telco strategic account team.   In this role he is involved in various aspects of helping a Tier 1 telco evolve towards a flexible infrastructure capable of delivering on the full promises of 5G. 
Prior to May 2016 Alex was a Senior Principal Engineer at InterDigital, leading the company’s research and development activities in the area wireless internet evolution.  Since joining InterDigital in 1999, he has been involved in a wide range of projects, including leadership of 3G modem ASIC architecture, design of advanced wireless security systems, coordination of standards strategy in the cognitive networks space, development of advanced IP mobility and heterogeneous access technologies and development of new content management techniques for the mobile edge. 
Alex earned his B.S.E.E. Summa Cum Laude from The Cooper Union, S.M. in Electrical Engineering and Computer Science from the Massachusetts Institute of Technology, and Ph.D. in Electrical Engineering from Princeton University. He held a visiting faculty appointment at WINLAB, Rutgers University, where he collaborated on research in cognitive radio, wireless security, and future mobile Internet.   He served as the Vice-Chair of the Services Working Group at the Small Cells Forum.  Alex is an inventor of over 130 granted U.S. patents, and has been awarded numerous awards for Innovation at InterDigital.
\end{IEEEbiography}

\end{document}